\documentclass[11pt]{article}
\usepackage{amsmath,amssymb,amsthm} 
\usepackage{bbm}
\usepackage{stmaryrd}
\usepackage{graphicx}
\usepackage{paralist}
\usepackage{hyperref}
\usepackage{color}

\setdefaultenum{(i)}{}{}{}
\usepackage[round]{natbib}
\usepackage[margin=1in]{geometry}

\newcommand\ve{\varepsilon}

\DeclareMathAlphabet\scr{U}{scr}{m}{n}
\SetMathAlphabet\scr{bold}{U}{scr}{b}{n}
  \DeclareFontFamily{U}{scr}{\skewchar\font'177}%
  \DeclareFontShape{U}{scr}{m}{n}{<-6>rsfs5<6-8>rsfs7<8->rsfs10}{}%
  \DeclareFontShape{U}{scr}{b}{n}{<-6>rsfs5<6-8>rsfs7<8->rsfs10}{}%

\newtheorem{theorem}{Theorem}[section]

\newtheorem{corollary}[theorem]{Corollary}

\newtheorem{definition}[theorem]{Definition}
\newtheorem{lemma}[theorem]{Lemma}
\newtheorem{proposition}[theorem]{Proposition}

\theoremstyle{definition}

\newcommand{\tilmu}{{\tilde\mu}}

\newcommand{\tilsigma}{{\tilde\sigma}}

\newcommand{\esp}[2][E]{#1\left[#2\right]}
\newcommand{\cale}{\mathcal E}

\newcommand{\eps}{\varepsilon}
\newcommand{\nada}[1]{}

\numberwithin{equation}{section}

\renewenvironment{thebibliography}[1]{%
\begin{oldthebibliography}{#1}%
\setlength{\baselineskip}{.9em}
\linespread{.9}
\small
\setlength{\parskip}{0ex}%
\setlength{\itemsep}{.1em}%
}%
{%
\end{oldthebibliography}%
}

\DeclareMathOperator{\lip}{LiPr}
\DeclareMathOperator{\sht}{ShTu}
\DeclareMathOperator{\wet}{WeTu}
\DeclareMathOperator{\cer}{ESR}

\begin{document}

\title{\vspace{-2.5cm}Transaction Costs, Trading Volume, \\and the Liquidity Premium\footnote{For helpful comments, we thank Maxim Bichuch, George Constantinides, Ale{\v{s}} {\v{C}}ern{\'y}, Mark Davis, Ioannis Karatzas, Ren Liu, Marcel Nutz, Scott Robertson, Johannes Ruf, Mihai Sirbu, Mete Soner, Gordan Zitkovi\'c, and seminar participants at Ascona, MFO Oberwolfach, Columbia University, Princeton University, University of Oxford, CAU Kiel, London School of Economics, University of Michigan, TU Vienna, and the ICIAM meeting in Vancouver. We are also very grateful to two anonymous referees for numerous -- and amazingly detailed -- remarks and suggestions.}}

\author{Stefan Gerhold
\thanks{Technische Universit\"at Wien, Institut f\"ur Wirtschaftsmathematik, Wiedner Hauptstrasse 8-10, A-1040 Wien, Austria, 
email \texttt{sgerhold@fam.tuwien.ac.at}. Partially supported by the Austrian Federal Financing Agency (FWF) and the Christian-Doppler-Gesellschaft (CDG).}
\and
Paolo Guasoni
\thanks{Boston University, Department of Mathematics and Statistics, 111 Cummington Street, Boston, MA 02215, USA, and Dublin City University, School of Mathematical Sciences, Glasnevin, Dublin 9, Ireland, email \texttt{guasoni@bu.edu}.
Partially supported by the ERC (278295), NSF (DMS-0807994, DMS-1109047), SFI (07/MI/008, 07/SK/M1189, 08/SRC/FMC1389), and FP7 (RG-248896).}
\and
Johannes Muhle-Karbe
\thanks{Corresponding author. ETH Z\"urich, Departement Mathematik, R\"amistrasse 101, CH-8092, Z\"urich, Switzerland, and Swiss Finance Institute, email \texttt{johannes.muhle-karbe@math.ethz.ch}. Partially supported by the National Centre of Competence in Research ``Financial Valuation and Risk Management'' (NCCR FINRISK), Project D1 (Mathematical Methods in Financial Risk Management), of the Swiss National Science Foundation (SNF).}
\and
Walter Schachermayer
\thanks{Universit\"at Wien, Fakult\"at f\"ur Mathematik, Nordbergstrasse 15, A-1090 Wien, Austria, email \texttt{walter.schachermayer@univie.ac.at}. Partially supported by the Austrian Science Fund (FWF) under grant P19456, the European Research Council (ERC) under grant FA506041, the Vienna Science and Technology Fund (WWTF) under grant MA09-003, and by the Christian-Doppler-Gesellschaft (CDG).}
}

\maketitle 

\vspace{-1cm}
\begin{abstract}
In a market with one safe and one risky asset, an investor with a long horizon, constant investment opportunities, and constant relative risk aversion trades with small proportional transaction costs. 
We derive explicit formulas for the optimal investment policy, its implied  welfare, liquidity premium, and trading volume. At the first order, the liquidity premium equals the spread, times share turnover, times a universal constant. Results are robust to consumption and finite horizons. 
We exploit the equivalence of the transaction cost market to another frictionless market, with a shadow risky asset, in which investment opportunities are stochastic. The shadow price is also found explicitly.
\end{abstract}




\section{Introduction}

If risk aversion and investment opportunities are constant --- and frictions are absent --- investors should hold a constant mix of safe and risky assets \citep*{markowitz.52,merton.69,MR0456373}. Transaction costs substantially  change this statement, casting some doubt on its far-reaching implications.\footnote{\citet*{constantinides.86} finds that ``transaction costs have a first-order effect on the assets' demand.''
\citet*{liu2002optimal} note that ``even small transaction costs lead to dramatic changes in the optimal behavior for an investor: from continuous trading to virtually buy-and-hold strategies.'' \citet*{luttmer.96} shows how  small transaction costs help resolve asset pricing puzzles.} 
Even the small spreads that are present in the most liquid markets entail wide oscillations in portfolio weights, which imply variable risk premia. 

This paper studies a tractable benchmark of portfolio choice under transaction costs, with constant investment opportunities, summarized by a safe rate $r$, and a risky asset with volatility $\sigma$ and expected excess return $\mu>0$, which trades at a bid (selling) price $(1-\ve)S_t$ equal to a constant fraction $(1-\eps)$ of the ask (buying) price $S_t$.
Our analysis is based on the model of \citet*{dumas.luciano.91}, which concentrates on long-run asymptotics to gain in tractability. In their framework, we find explicit solutions for the optimal policy, welfare, liquidity premium\footnote{ That is, the amount of excess return the investor is ready to forgo to trade the risky asset without transaction costs.} and trading volume, in terms of model parameters, and of an additional quantity, the \emph{gap}, identified as the solution to a scalar equation. For all these quantities, we derive closed-form asymptotics, in terms of model parameters only, for small transaction costs.

We uncover novel relations among the liquidity premium, trading volume, and transaction costs. First, we show that share turnover ($\sht$), the liquidity premium ($\lip$), and the bid-ask spread $\eps$ satisfy the following asymptotic relation:
\begin{equation*}
\lip \approx \frac34 \eps \sht.
\end{equation*}
This relation is universal, as it involves neither market nor preference parameters. Also, because it links the liquidity premium, which is unobservable, with spreads and share turnover, which are observable, this relation can help estimate the liquidity premium using data on trading volume. 

Second, we find that the liquidity premium behaves very differently in the presence of leverage. In the no-leverage regime, the liquidity premium is an order of magnitude smaller than the spread \citep{constantinides.86}, as unlevered investors respond to transaction costs by trading infrequently. With leverage, however, the liquidity premium increases quickly, because rebalancing a levered position entails high transaction costs, even under the optimal trading policy.

Third, we obtain the first continuous-time benchmark for trading volume, with explicit formulas for share and wealth turnover. 
Trading volume is an elusive quantity for frictionless models, in which turnover is typically infinite in any time interval.\footnote{The empirical literature has long been aware of this theoretical vacuum: \citet*{gallant1992stock} reckon that ``The intrinsic difficulties of specifying plausible, rigorous, and implementable models of volume and prices are the reasons for the informal modeling approaches commonly used.'' \citet*{lo2000trading}  note that ``although most models of asset markets have focused on the behavior of returns [...] their implications for trading volume have received far less attention.''} 
In the absence of leverage, our results imply low trading volume compared to the levels observed in the market. Of course, our model can only explain trading  generated by portfolio rebalancing, and not by other motives such as market timing, hedging, and life-cycle investing.

Moreover, welfare, the liquidity premium, and trading volume depend on the market parameters ($\mu,\sigma$) only through the mean-variance ratio $\mu/\sigma^2$ if measured in \emph{business time}, that is, using a clock that ticks at the speed of the market's variance $\sigma^2$. In usual calendar time, all these quantities are in turn multiplied by the variance $\sigma^2$. 

Our main implication for portfolio choice is that a symmetric, stationary policy is optimal for a long horizon, and it is robust, at the first order, both to intermediate consumption, and to a finite horizon. Indeed, we show that the no-trade region is perfectly symmetric with respect to the Merton proportion $\pi_*=\mu/\gamma\sigma^2$, \emph{if} trading boundaries are expressed with trading prices, that is, if the buy boundary $\pi_-$ is computed from the ask price, and the sell boundary $\pi_+$ from the bid price. 

Since in a frictionless market the optimal policy is independent both of intermediate consumption and of the horizon \citep{MR0456373}, our results entail that these two features are robust to small frictions. However plausible these conclusions may seem, the literature so far has offered diverse views on these issues (cf. \citet*{MR1080472,dumas.luciano.91,liu2002optimal}). 
More importantly, robustness to the horizon implies that the long-horizon approximation, made for the sake of tractability, is reasonable and relevant. For typical parameter values, we see that our optimal strategy is nearly optimal already for horizons as short as two years. 

A key idea for our results --- and for their proof --- is the equivalence between a market with transaction costs and constant investment opportunities, and another \emph{shadow} market, without transaction costs, but with stochastic investment opportunities driven by a state variable. This state variable is the ratio between the investor's risky and safe weights, which tracks the location of the portfolio within the trading boundaries, and affects both the volatility and the expected return of the shadow risky asset. 

In this paper, using a shadow price has two related advantages over alternative methods: first, it allows us to tackle the issue of verification with duality methods developed for frictionless markets. These duality methods in turn yield the finite-horizon bounds in Theorem \ref{th:finhor} below, which measure the performance of long-run policies over a given horizon -- an issue that is especially important when an asymptotic objective funcion is used.
The shadow price method was applied successfully by \citet*{MR2676941,gerhold.al.10b, gerhold.al.10a} for logarithmic utility, and this paper brings this approach to power utility, which allows to understand how optimal policies, welfare, liquidity premia and trading volume depend on risk aversion. The recent papers of \citet*{prokaj.12,choi.al.12} consider power utility from consumption in an infinite horizon.

The paper is organized as follows: Section 2 introduces the portfolio choice problem and states the main results. The model's main implications are discussed in Section 3, and the main results are derived heuristically in Section 4. Section 5 concludes, and all proofs are in the appendix. 

\section{Model and Main Result}

Consider a market with a safe asset earning an interest rate $r$, i.e. $S^0_t=e^{rt}$, and a risky asset, trading at ask (buying) price $S_t$ following geometric Brownian motion,
\begin{equation*}
dS_t/S_t=(\mu+r) dt+\sigma dW_t.
\end{equation*}
Here, $W_t$ is a standard Brownian motion, $\mu>0$ is the expected excess return,\footnote{A negative excess return leads to a similar treatment, but entails buying as prices rise, rather than fall. For the sake of clarity, the rest of the paper concentrates on the more relevant case of a positive $\mu$.} and $\sigma>0$ is the {volatility}. The corresponding bid (selling) price is $(1-\ve)S_t$, where $\ve \in (0,1)$ represents the relative bid-ask spread.

A self-financing \emph{trading strategy} is a two-dimensional, predictable process $(\varphi^0_t,\varphi_t)$ of finite variation, such that $\varphi^0_t$ and $\varphi_t$ represent the number of units in the safe and risky asset at time $t$, and the initial number of units is $(\varphi^0_{0^-},\varphi_{0^-})=(\xi^0,\xi) \in \mathbb{R}^2_+\backslash \{0,0\}$. Writing $\varphi_t = \varphi^{\uparrow}_t-\varphi^{\downarrow}_t$ as the difference between the cumulative number of shares bought ($\varphi^{\uparrow}_t$) and sold ($\varphi^{\downarrow}_t$) by time $t$, the \emph{self-financing condition} relates the dynamics of $\varphi^0$ and $\varphi$ via
\begin{equation}\label{eq:selffinancing}
d\varphi^0_t =  -\frac{S_t}{S^0_t} d\varphi_t^{\uparrow}+ (1-\ve)\frac{S_t}{S^0_t} d\varphi^{\downarrow}_t .
\end{equation}
As in \citet*{dumas.luciano.91}, the investor maximizes the equivalent safe rate of power utility, an optimization objective that also proved useful with constraints on leverage \citep*{grossman1992optimal} and drawdowns \citep*{grossman1993optimal}.

\begin{definition}\label{def:admiss}
A trading strategy $(\varphi_t^0,\varphi_t)$ is \emph{admissible} if its liquidation value is positive, in that:
\begin{equation*}
\Xi^\varphi_t=\varphi^0_t S^0_t+(1-\ve)S_t\varphi_t^+ -\varphi_t^- S_t\ge 0, \qquad\text{a.s. for all }t\ge 0.
\end{equation*}
An admissible strategy $(\varphi_t^0,\varphi_t)$ is \emph{long-run optimal} if it maximizes the \emph{equivalent safe rate}
 \begin{equation}\label{eq:longrun}
\liminf_{T \to \infty} \frac{1}{T}\log
E\left[(\Xi^\varphi_T)^{1-\gamma}\right]^{\frac1{1-\gamma}}
 \end{equation}
over all admissible strategies, where $1\ne \gamma>0$ denotes the investor's relative risk aversion.\footnote{The limiting case $\gamma\rightarrow 1$ corresponds to logarithmic utility, studied by \citet*{MR942619}, \citet*{akian2001dynamic}, and \citet*{gerhold.al.10a}. Theorem \ref{th:main} remains valid for logarithmic utility setting $\gamma=1$.}
 \end{definition}

Our main result is the following:
\begin{theorem}\label{th:main}
An investor with constant relative risk aversion $\gamma>0$ trades to maximize \eqref{eq:longrun}. Then, for small transaction costs $\ve>0$:
\begin{enumerate}[i)]
\item 
\emph{(Equivalent Safe Rate)}\\
For the investor, trading the risky asset with transaction costs is equivalent to leaving all wealth in a hypothetical safe asset, which pays the higher  \emph{equivalent safe rate}:
\begin{equation}
\cer=r+\frac{\mu^2-\lambda^2}{2\gamma\sigma^2}, 
\end{equation}
where the \emph{gap} $\lambda$ is defined in $iv)$ below.
\item
\emph{(Liquidity Premium)}\\
Trading the risky asset with transaction costs is equivalent to trading a hypothetical asset, at no transaction costs, with the same volatility $\sigma$, but with lower expected excess return $\sqrt{\mu^2-\lambda^2}$. Thus, the \emph{liquidity premium} is
\begin{equation}
\lip = \mu-\sqrt{\mu^2-\lambda^2}.
\end{equation}

\item 
\emph{(Trading Policy)}\\
It is optimal to keep the fraction of wealth held in the risky asset within the buy and sell boundaries
\begin{equation}\label{portfolios}
\pi_-=\frac{\mu-\lambda}{\gamma \sigma^2},
\qquad
\pi_+=\frac{\mu+\lambda}{\gamma \sigma^2},
\end{equation}
where the risky weights $\pi_-$ and $\pi_+$ are computed with ask and bid prices, respectively.\footnote{This optimal policy is not necessarily unique, in that its long-run performance is also attained by trading arbitrarily for a finite time, and then switching to the above policy. However, in related frictionless models, as the horizon increases, the optimal (finite-horizon) policy converges to a stationary policy, such as the one considered here (see, e.g., \citet*{dybvig.al.99}). \citet*{dai.yi.09} obtain similar results in a model with proportional transaction costs, formally passing to a stationary version of their control problem PDE.}

\item 
\emph{(Gap)}\\
For $\mu/\gamma\sigma^2 \neq 1$,  the constant $\lambda \geq 0$ is the unique value for which the solution of the initial value problem
\begin{align*}
& w'(x)+(1-\gamma)w(x)^2+\left(\frac{2\mu}{\sigma^2}-1\right)w(x)-
\gamma \left(\frac{\mu-\lambda}{\gamma \sigma^2}\right)
\left(\frac{\mu+\lambda}{\gamma \sigma^2}\right)
= 0,
\\
& w(0) = \frac{\mu-\lambda}{\gamma \sigma^2}
\end{align*}
also satisfies the terminal value condition:
\begin{equation*}
w\left(\log \left(\frac{u(\lambda)}{l(\lambda)}\right)\right) = \frac{\mu+\lambda}{\gamma \sigma^2},
\qquad\text{where}\qquad
\frac{u(\lambda)}{l(\lambda)} = \frac{1}{(1-\eps)}\frac{(\mu+\lambda)(\mu-\lambda-\gamma \sigma^2)}
{(\mu-\lambda)(\mu+\lambda-\gamma \sigma^2)}.
\end{equation*}
In view of the explicit formula for $w(x,\lambda)$ in Lemma~\ref{lem:riccati} below, this is a scalar equation for $\lambda$. For $\mu/\gamma\sigma^2=1$, the gap $\lambda$ vanishes.

\item
\emph{(Trading Volume)}\\
Let $\mu \neq \sigma^2/2$.\footnote{The corresponding formulas for $\mu=\sigma^2/2$ are similar but simpler, compare Corollary \ref{cor:trading} and Lemma \ref{lem:locallimits}.} Then \emph{share turnover}, defined as shares traded $d\|\varphi\|_t=d\varphi^\uparrow_t+d\varphi^\downarrow_t$ divided by shares held $|\varphi_t|$, has the long-term average
\begin{align*}
\sht =
\lim_{T\rightarrow\infty}\frac1T\int_0^T \frac{d\|\varphi\|_t}{|\varphi_t|}=
\frac{\sigma^2}{2}\left(\frac{2\mu}{\sigma^2}-1\right)
\left(
\frac{1-\pi_-}{(u(\lambda)/l(\lambda))^{\frac{2\mu}{\sigma^2}-1}-1} -\frac{1-\pi_+}{(u(\lambda)/l(\lambda))^{1-\frac{2\mu}{\sigma^2}}-1}
\right).
\end{align*}
\emph{Wealth turnover}, defined as wealth traded divided by wealth held, has long term-average\footnote{The number of shares is written as the difference $\varphi_t=\varphi^{\uparrow}_t-\varphi^{\downarrow}_t$ of the cumulative shares bought (resp.\ sold), and wealth is evaluated at trading prices, i.e., at the bid price $(1-\ve)S_t$ when selling, and at the ask price $S_t$ when buying.} 
\begin{align*}
\wet &= \lim_{T\to \infty} \frac{1}{T}\left(\int_0^T \frac{(1-\ve)S_t d\varphi^{\downarrow}_t}{\varphi^0_tS^0_t+\varphi_t (1-\ve)S_t}+\int_0^T \frac{S_t d\varphi^{\uparrow}_t}{\varphi^0_tS^0_t+\varphi_t S_t}\right)\\
 &= 
\frac{\sigma^2}{2}
\left(\frac{2\mu}{\sigma^2}-1\right)
\left(
\frac{\pi_-\left(1-\pi_-\right)}{(u(\lambda)/l(\lambda))^{\frac{2\mu}{\sigma^2}-1}-1}
-
\frac{\pi_+\left(1-\pi_+\right)}{(u(\lambda)/l(\lambda))^{1-\frac{2\mu}{\sigma^2}}-1}
\right)
.
\end{align*}

\item
\emph{(Asymptotics)}\\
Setting $\pi_*=\mu/\gamma\sigma^2$, the following  expansions in terms of the bid-ask spread $\eps$ hold:\footnote{Algorithmic calculations can deliver terms of arbitrarily high order.}
\begin{align}
\label{eq:gapexp}
\lambda &= \gamma\sigma^2 
\left(\frac{3}{4\gamma} \pi_*^2\left(1-\pi_*\right)^2\right)^{1/3} \ve^{1/3}
+ O(\ve).
\\
\cer &= 
r+\frac{\mu^2}{2\gamma\sigma^2}-
\frac{\gamma\sigma^2}{2}
\left(\frac{3}{4\gamma} \pi_*^2\left(1-\pi_*\right)^2\right)^{2/3}\ve^{2/3} +O(\ve^{4/3}).
\\
\lip &= 
\frac{\mu}{2\pi_*^2}\left(\frac{3}{4\gamma} \pi_*^2\left(1-\pi_*\right)^2\right)^{2/3} \ve^{2/3}+O(\ve^{4/3}).\\
\pi_{\pm} &= \pi_* \pm 
\left(\frac{3}{4\gamma} \pi_*^2\left(1-\pi_*\right)^2\right)^{1/3} \ve^{1/3} +O(\ve).\\
\sht &= 
\frac{\sigma^2}2 
(1-\pi_*)^2\pi_* 
\left(\frac{3}{4\gamma} \pi_*^2\left(1-\pi_*\right)^2\right)^{-1/3}\ve^{-1/3}
  + O(\ve^{1/3}).
\\
\label{eq:wetexp}
\wet &= 
\frac{2\gamma\sigma^2}{3} \left(\frac{3}{4\gamma} \pi_*^2(1-\pi_*)^2\right)^{2/3} \ve^{-1/3}+O(\ve^{1/3}).
\end{align}

\end{enumerate}
\end{theorem}

In summary, our optimal trading policy, and its resulting welfare, liquidity premium, and trading volume are all simple functions of investment opportunities ($r$, $\mu$, $\sigma$), preferences ($\gamma$), and the gap $\lambda$. The gap does not admit an explicit formula in terms of the transaction cost parameter $\ve$, but is determined through the implicit relation in $iv)$, and has the asymptotic expansion in $vi)$, from which all other asymptotic expansions follow through the explicit formulas.

The frictionless markets with constant investment opportunities in items $i)$ and $ii)$ of Theorem~\ref{th:main} are equivalent to the market with transaction costs in terms of equivalent safe rates. Nevertheless, the corresponding optimal policies are very different, requiring no or incessant rebalancing in the frictionless markets of $i)$ and $ii)$, respectively, whereas there is finite positive trading volume in the market with transaction costs. 

By contrast, the shadow price, which is key in the derivation of our results, is a fictitious risky asset, with price evolving within the bid-ask spread, for which the corresponding frictionless market is equivalent to the transaction cost market in terms of both welfare \emph{and} the optimal policy:

\begin{theorem}\label{thm:shadow}
The policy in Theorem \ref{th:main} $iii)$ and the equivalent safe rate in Theorem \ref{th:main} $i)$ are also optimal for a frictionless asset with \emph{shadow price} $\tilde{S}$, which always lies within the bid-ask spread, and coincides with the trading price at times of trading for the optimal policy. The shadow price satisfies
\begin{equation}
d\tilde S_t/\tilde S_t = (\tilde{\mu}(\Upsilon_t)+r) dt + \tilde{\sigma}(\Upsilon_t) dW_t,
\end{equation}
for the deterministic functions $\tilde{\mu}(\cdot)$ and $\tilde{\sigma}(\cdot)$ given explicitly in Lemma \ref{lem:dynamics}. The state variable $\Upsilon_t=\log(\varphi_t S_t/(l(\lambda)\varphi^0_t S^0_t))$ represents the logarithm of the ratio of risky and safe positions, which follows a Brownian motion with drift, reflected to remain in the interval $[0,\log(u(\lambda)/l(\lambda))]$, i.e.,
\begin{equation}
d\Upsilon_t=(\mu-\sigma^2/2)dt+\sigma dW_t +dL_t-dU_t.
\end{equation}
\end{theorem}
Here, $L_t$ and $U_t$ are increasing processes, proportional to the cumulative purchases and sales, respectively (cf.\ \eqref{eq:philu} below).
In the interior of the no-trade region, that is, when $\Upsilon_t$ lies in $(0,\log (u(\lambda)/l(\lambda)) )$, the numbers of units of the safe and risky asset are constant, and the state variable $\Upsilon_t$ follows Brownian motion with drift. As $\Upsilon_t$ reaches the boundary of the no-trade region, buying or selling takes place as to keep it within $[0,\log (u(\lambda)/l(\lambda))]$.

In view of Theorem \ref{thm:shadow}, trading with constant investment opportunities and proportional transaction costs is equivalent to trading in a fictitious frictionless market with stochastic investment opportunities, which vary with the location of the investor's portfolio in the no-trade region.

\section{Implications}

\subsection{Trading Strategies}
Equation \eqref{portfolios} implies that trading boundaries are symmetric around the frictionless Merton proportion $\pi_*=\mu/\gamma\sigma^2$. At first glance, this seems to contradict previous studies (e.g., \citet*{liu2002optimal}, \citet*{MR1284980}), which emphasize how these boundaries are asymmetric, and may even fail to include the Merton proportion. These papers employ a common reference price (the average of the bid and ask prices) to evaluate both boundaries. By contrast, we express trading boundaries using trading prices (i.e., the ask price for the buy boundary, and the bid price for the sell boundary). This simple convention unveils the natural symmetry of the optimal policy, and explains asymmetries as figments of notation -- even in their models.
To see this, denote by $\pi_-'$ and $\pi_+'$ the buy and sell boundaries in terms of the ask price. These papers prove the bounds (\citet*[equations (11.4) and (11.6)]{MR1284980} in an infinite-horizon model with consumption and \citet*[equations (22), (23)]{liu2002optimal} in a finite-horizon model)
\begin{equation}
\pi_-'< \frac{\mu}{\gamma \sigma^2}
\quad\text{and}\quad
\frac{\mu}{\gamma \sigma^2(1-\varepsilon)+\varepsilon \mu} <
\pi_+' < \frac{\mu}{\frac12\gamma \sigma^2(1-\varepsilon)+\varepsilon \mu}\ .
\end{equation}
With trading prices (i.e., substituting $\pi_-=\pi_-'$ and $\pi_+=\frac{1-\varepsilon}{1-\varepsilon \pi_+'}\pi_+'$) these bounds become
\begin{equation}
\pi_-< \frac{\mu}{\gamma \sigma^2} <
\pi_+ < 2\frac{\mu}{\gamma \sigma^2}\ ,
\end{equation}
whence the Merton proportion always lies between $\pi_-$ and $\pi_+$.

To understand the robustness of our optimal policy to intermediate consumption, we compare our trading boundaries with those obtained by \citet*{MR1080472} and \cite{MR1284980} in the consumption model of \citet*{magill.constantinidis.76}. The asymptotic expansions of \citet*{MR2048827} make this comparison straightforward.

With or without consumption, the trading boundaries coincide at the first-order. This fact has a clear economic interpretation: the separation between consumption and investment, which holds in a frictionless model with constant investment opportunities, is a robust feature of frictionless models, because it still holds, \emph{at the first order}, even with transaction costs. Put differently, if investment opportunities are constant, consumption has only a second order effect for investment decisions, in spite of the large no-trade region implied by transaction costs. Figure \ref{fig:trading_bounds} shows that our bounds are very close to those obtained in the model of \citet*{MR1080472} for bid-ask spreads below 1\%, but start diverging for larger values. 

\begin{figure}
\includegraphics[height=1.9in]{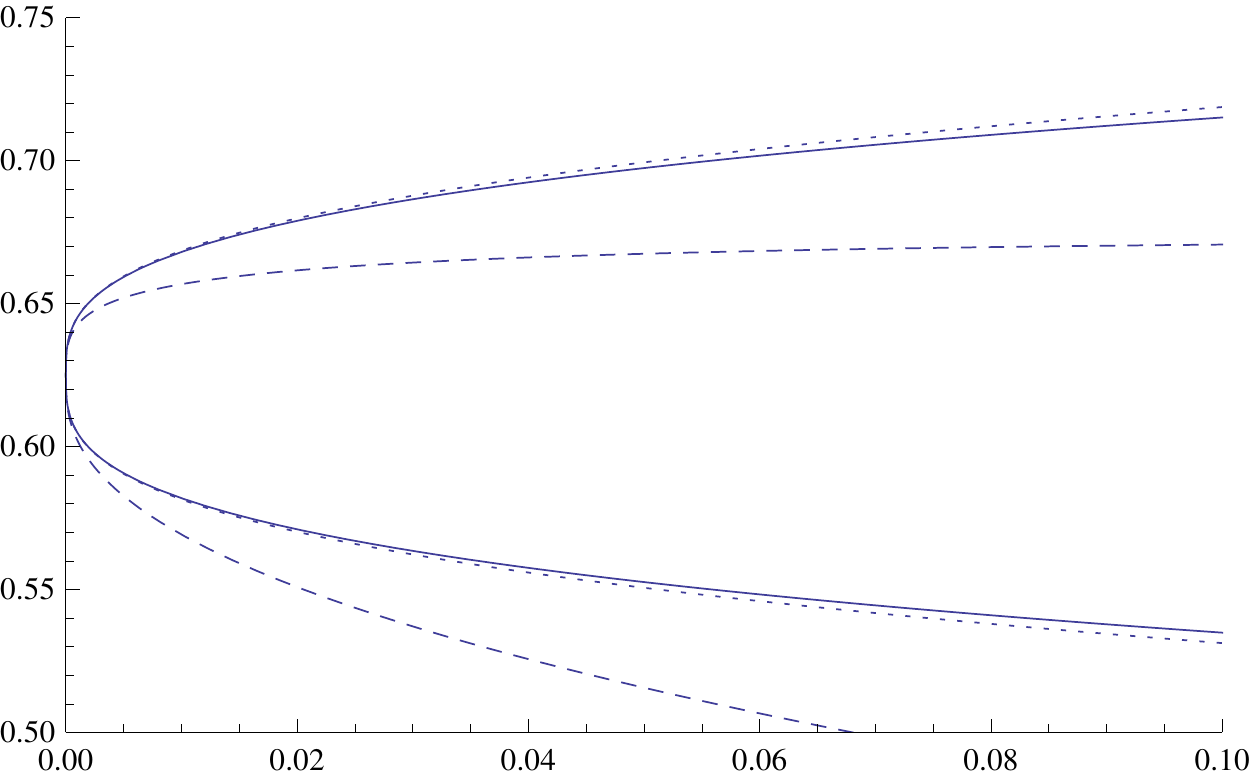}
\includegraphics[height=1.9in]{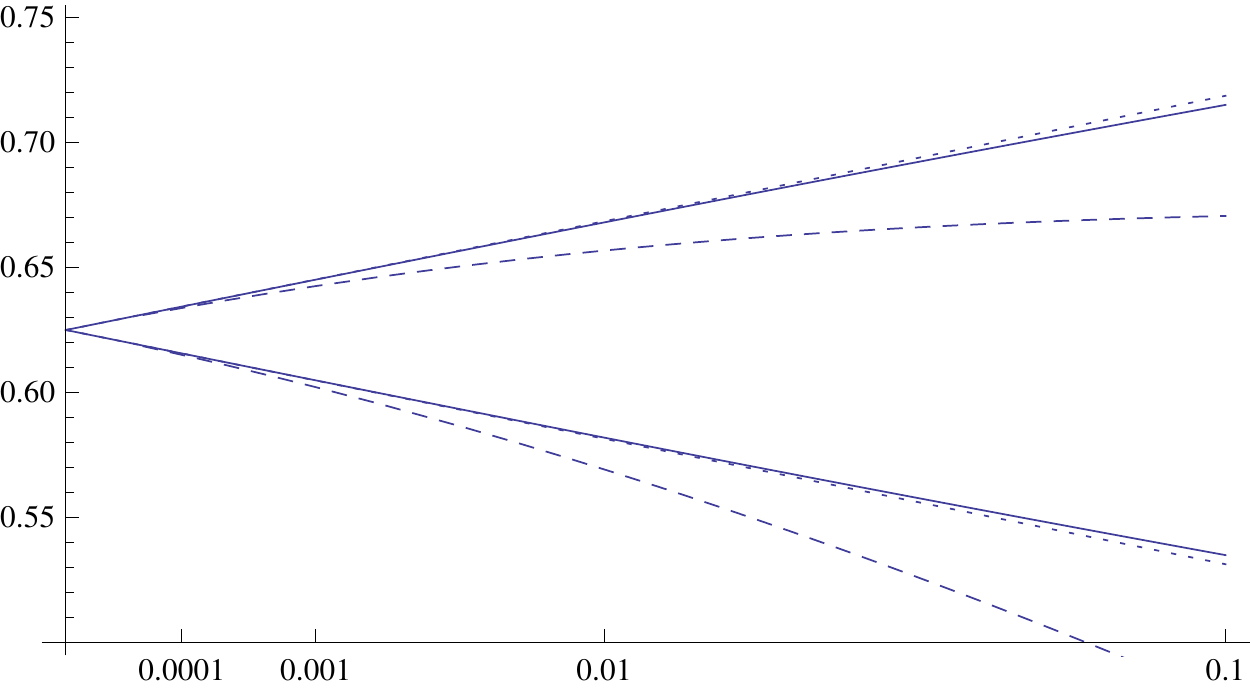}
\caption{\label{fig:trading_bounds}
Buy (lower) and sell (upper) boundaries (vertical axis, as risky weights) as functions of the spread $\varepsilon$, in linear scale (left panel) and cubic scale (right panel). 
The plot compares the approximate weights from the first term of the expansion (dotted), the exact optimal weights (solid), and the boundaries found by \citet*{MR1080472} in the presence of consumption (dashed). Parameters are $\mu=8\%, \sigma=16\%, \gamma=5$, and a zero discount rate for consumption (for the dashed curve).}
\end{figure}

\subsection{Business time and Mean-Variance Ratio}

In a frictionless market, the equivalent safe rate and the optimal policy are:
\[
\cer = r + \frac{1}{2\gamma} \left(\frac{\mu}{\sigma}\right)^2
\qquad\text{and}\qquad
\pi_* = \frac{\mu}{\gamma \sigma^2}.
\]
This rate depends only on the safe rate $r$ and the Sharpe ratio $\mu/\sigma$. Investors are indifferent between two markets with identical safe rates and Sharpe ratios, because both markets lead to the same set of payoffs, even though a payoff is generated by different portfolios in the two markets. By contrast, the optimal portfolio depends only on the mean-variance ratio $\mu/\sigma^2$.

With transaction costs, Equation \eqref{eq:gapexp} shows that the asymptotic expansion of the gap per unit of variance $\lambda/\sigma^2$ only depends on the mean-variance ratio $\mu/\sigma^2$. Put differently, holding the mean-variance ratio $\mu/\sigma^2$ constant, the expansion of $\lambda$ is linear in $\sigma^2$. In fact, not only the expansion but also the exact quantity has this property, since $\lambda/\sigma^2$ in $iv)$ only depends on $\mu/\sigma^2$. 

Consequently, the optimal policy in $iii)$ only depends on the mean-variance ratio $\mu/\sigma^2$, as in the frictionless case. The equivalent safe rate, however, no longer solely depends on the Sharpe ratio $\mu/\sigma$: investors are not indifferent between two markets with the same Sharpe ratio, because one market is more attractive than the other if it entails lower trading costs. As an extreme case, in one market it may be optimal lo leave all wealth in the risky asset, eliminating any need to trade. Instead, the formulas in $i)$, $ii)$, and $v)$ show that, like the gap per variance $\lambda/\sigma^2$, the equivalent safe rate, the liquidity premium, and both share and wealth turnover only depend on $\mu/\sigma^2$, when measured per unit of variance. The interpretation is that these quantities are proportional to business time $\sigma^2 t$ \citep*{ane2000order}, and the factor of $\sigma^2$ arises from measuring them in calendar time.

In the frictionless limit, the linearity in $\sigma^2$ and the dependence on $\mu/\sigma^2$ cancel, and the result depends on the Sharpe ratio alone. For example, the equivalent safe rate becomes\footnote{The other quantities are trivial: the gap and the liquidity premium become zero, while share and wealth turnover explode to infinity.}
\[
r+\frac{\sigma^2}{2\gamma} \left(\frac{\mu}{\sigma^2}\right)^2 =r+\frac{1}{2\gamma}\left(\frac{\mu}{\sigma}\right)^2 .
\]

\subsection{Liquidity Premium}

The liquidity premium \citep*{constantinides.86} is the amount of expected excess return the investor is ready to forgo to trade the risky asset without transaction costs, as to achieve the same equivalent safe rate. Figure \ref{fig:liquid_premium} plots the liquidity premium against the spread $\eps$ (left panel) and risk aversion $\gamma$ (right panel). 

\begin{figure}[t]
\includegraphics[height=2in]{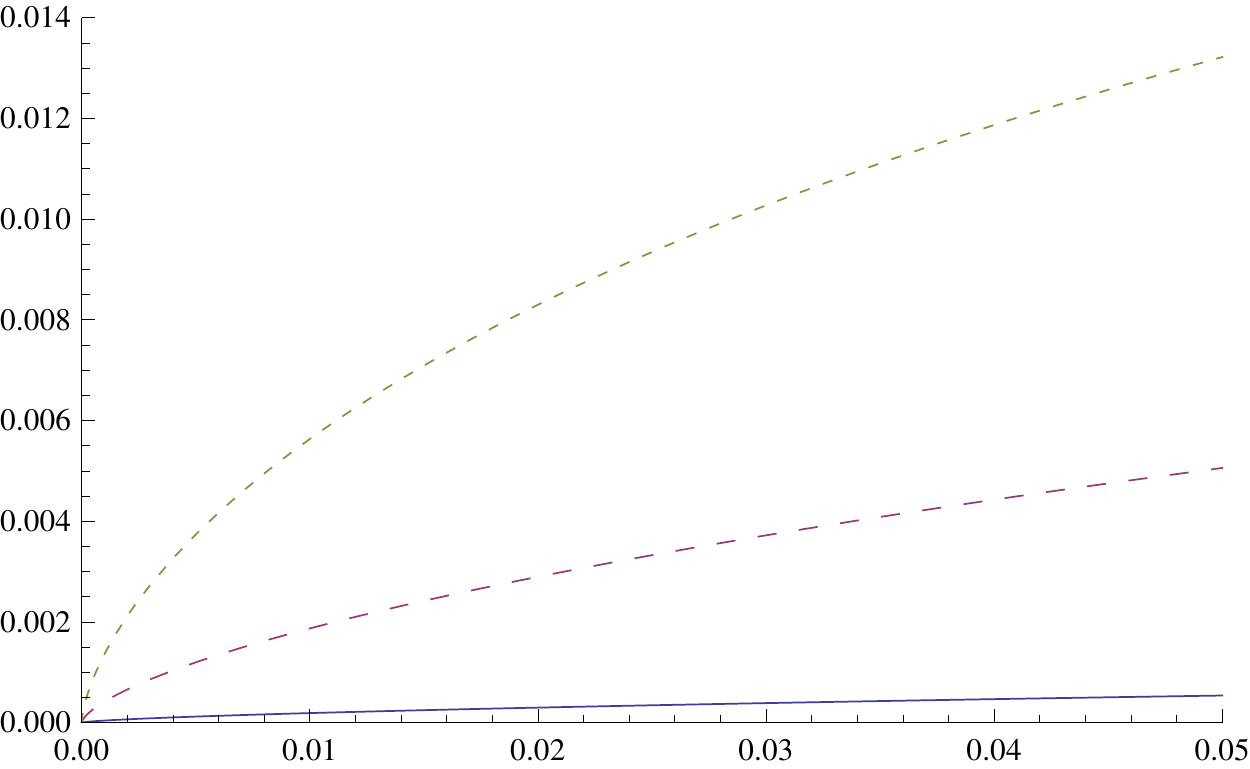}
\includegraphics[height=2in]{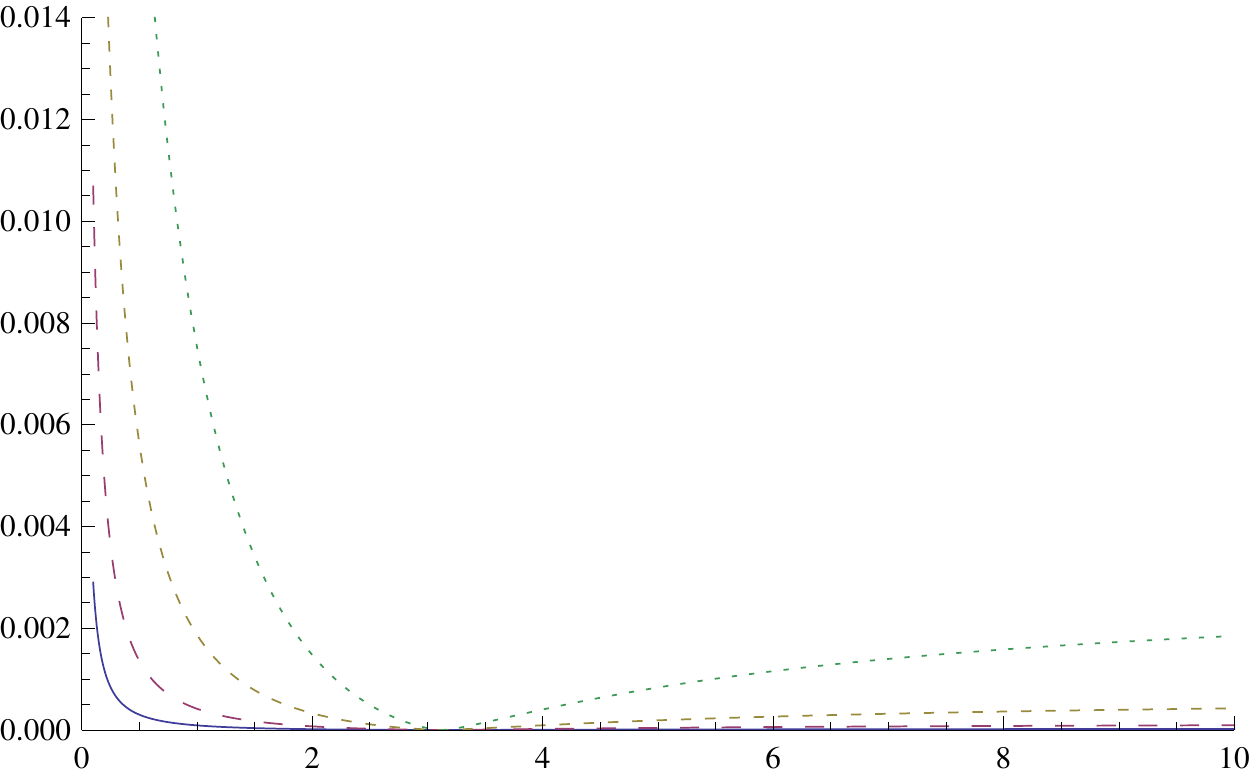}
\caption{\label{fig:liquid_premium}
Left panel: liquidity premium (vertical axis) against the spread $\eps$, for risk aversion $\gamma$ equal to $5$ (solid), $1$ (long dashed), and $0.5$ (short dashed). 
Right panel: liquidity premium (vertical axis) against risk aversion $\gamma$, for spread $\eps=0.01\%$ (solid), $0.1\%$ (long dashed), $1\%$ (short dashed), and $10\%$ (dotted). 
Parameters are $\mu=8\%$ and $\sigma=16\%$.}
\end{figure}

The liquidity premium is exactly zero when the Merton proportion $\pi_*$ is either zero or one. In these two limit cases, it is optimal not to trade at all, hence no compensation is required for the costs of trading.
The liquidity premium is relatively low in the regime of no leverage ($0<\pi_*<1$), corresponding to $\gamma>\mu/\sigma^2$, confirming the results of \citet*{constantinides.86}, who reports liquidity premia one order of magnitude smaller than trading costs.

The leverage regime ($\gamma<\mu/\sigma^2$), however, shows a very different picture. As risk aversion decreases below the full-investment level $\gamma=\mu/\sigma^2$, the liquidity premium increases rapidly towards the expected excess return $\mu$, as lower levels of risk aversion prescribe increasingly high leverage. The costs of rebalancing a levered position are high, and so are the corresponding liquidity premia.

The liquidity premium increases in spite of the increasing width of the no-trade region for larger leverage ratios. 
In other words, even as a less risk averse investor tolerates wider oscillations in the risky weight, this increased flexibility is not enough to compensate for the higher costs required to rebalance a more volatile portfolio.

\subsection{Trading Volume}

\begin{figure}[t]
\includegraphics[height=2in]{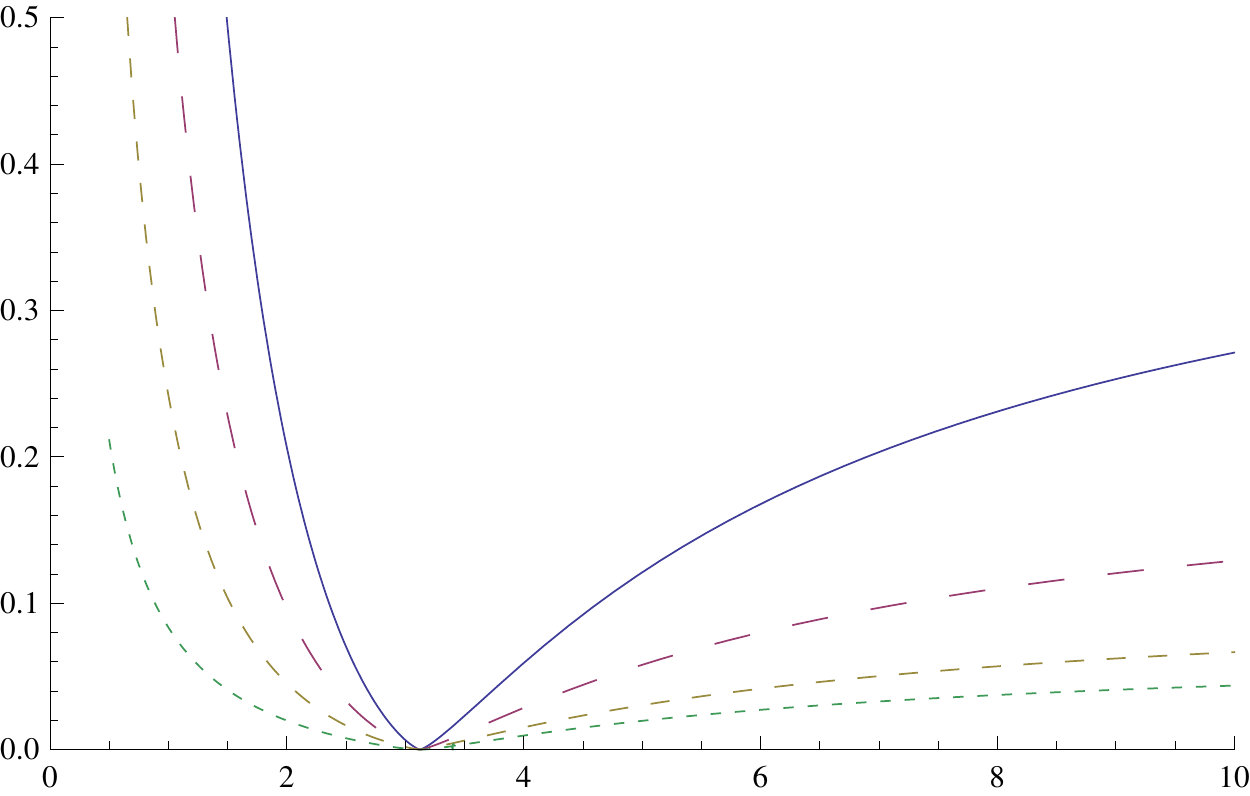}
\includegraphics[height=2in]{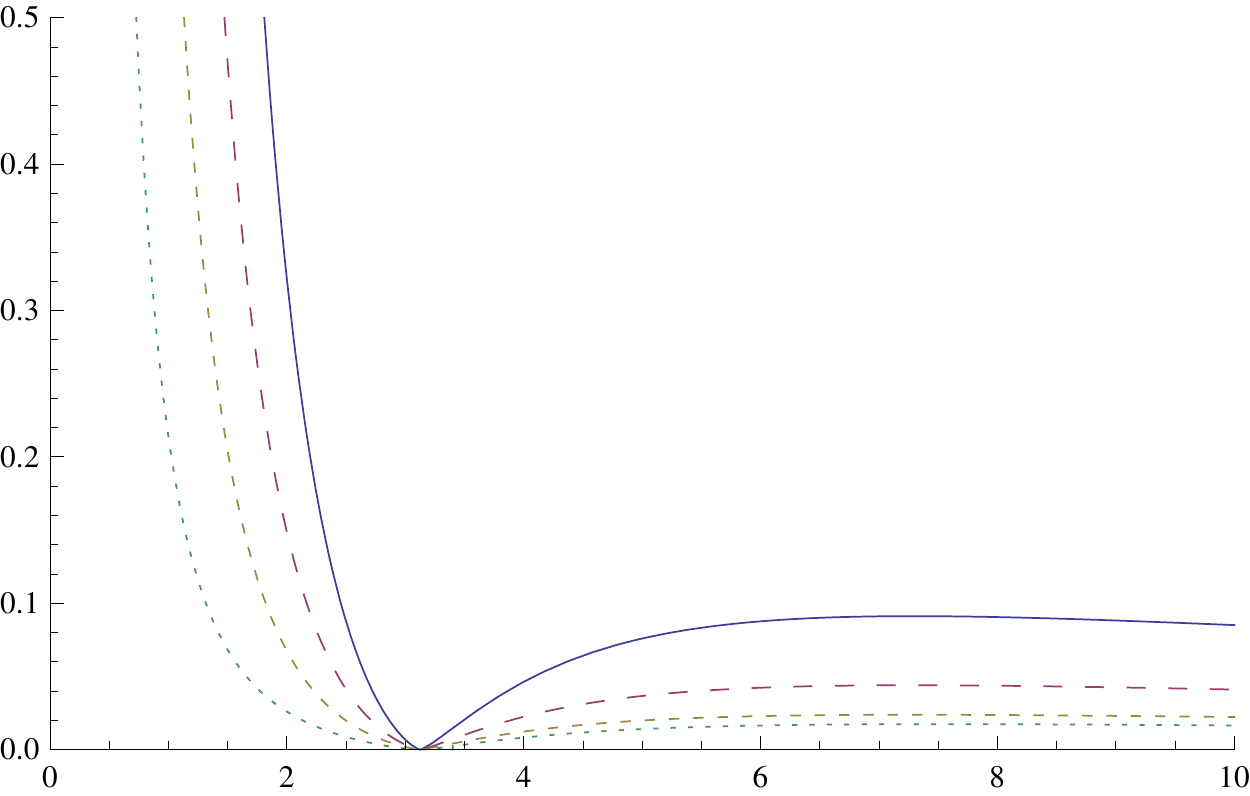}
\caption{\label{fig:turnover}
Trading volume (vertical axis, annual fractions traded), as share turnover (left panel) and wealth turnover (right panel), against risk aversion (horizontal axis), for spread $\eps=0.01\%$ (solid), $0.1\%$ (long dashed), $1\%$ (short dashed), and $10\%$ (dotted).
Parameters are $\mu=8\%$ and $\sigma=16\%$.}
\end{figure}

In the empirical literature (cf.\ \citet*{lo2000trading} and the references therein), the most common measure of trading volume is share turnover, defined as number of shares traded divided by shares held or, equivalently, as the value of shares traded divided by value of shares held. In our model, turnover is positive only at the trading boundaries, while it is null inside the no-trade region. Since turnover, on average, grows linearly over time, we consider the long-term average of share turnover per unit of time, plotted in Figure \ref{fig:turnover} against risk aversion. Turnover is null at the full-investment level $\gamma=\mu /\sigma^2$, as no trading takes place in this case. Lower levels of risk aversion generate leverage, and trading volume increases rapidly, like the liquidity premium.

Share turnover does not decrease to zero as the risky weight decreases to zero for increasing risk aversion $\gamma$. On the contrary, the first term in the asymptotic formula converges to a finite level.
This phenomenon arises because more risk averse investors hold less risky assets (reducing volume), but also rebalance more frequently (increasing volume). As risk aversion increases, neither of these effects prevails, and turnover converges to a finite limit.

To better understand these properties, consider wealth turnover, defined as the value of shares traded, divided by total wealth (not by the value of shares held).\footnote{Technically, wealth is valued at the ask price at the buying boundary, and at the bid price at the selling boundary.} 
Share and wealth turnover are qualitatively similar for low risk aversion, as the risky weight of wealth is larger, but they diverge as risk aversion increases and the risky weight declines to zero. Then, wealth turnover decreases to zero, whereas share turnover does not.

The levels of trading volume observed empirically imply very low values of risk aversion in our model. For example, \citet*{lo2000trading} report in the NYSE-AMEX an average \emph{weekly} turnover of 0.78\% between 1962-1996, which corresponds to an approximate annual turnover above 40\%. As Figure \ref{fig:turnover} shows, such a high level of turnover requires a risk aversion below 2, even for a very small spread of $\eps=0.01\%$. Such a value cannot be interpreted as risk aversion of a representative investor, because it would imply a leveraged position in the stock market, which is inconsistent with equilibrium.
This phenomenon intensifies in the last two decades. As shown by Figure \ref{fig:liq_turn_spr} turnover increases substantially from 1993 to 2010, with monthly averages of 20\% typical from 2007 on, corresponding to an annual turnover of over 240\%. 

The overall implication is that portfolio rebalancing can generate substantial trading volume, but the model explains the trading volume observed empirically only with low risk aversion and high leverage. In a numerical study with risk aversion of six and spreads of 2\%, \citet*{lynch2011explaining} also find that the resulting trading volume is too low, even allowing for labor income and predictable returns, and obtain a condition on the wealth-income ratio, under which the trading volume is the same order of magnitude reported by empirical studies. Our analytical results are consistent with their findings, but indicate that substantially higher volume can be explained with lower risk aversion, even in the absence of labor income.

\subsection{Volume, Spreads and the Liquidity Premium}

\begin{figure}[t]
\includegraphics[width=3.5in]{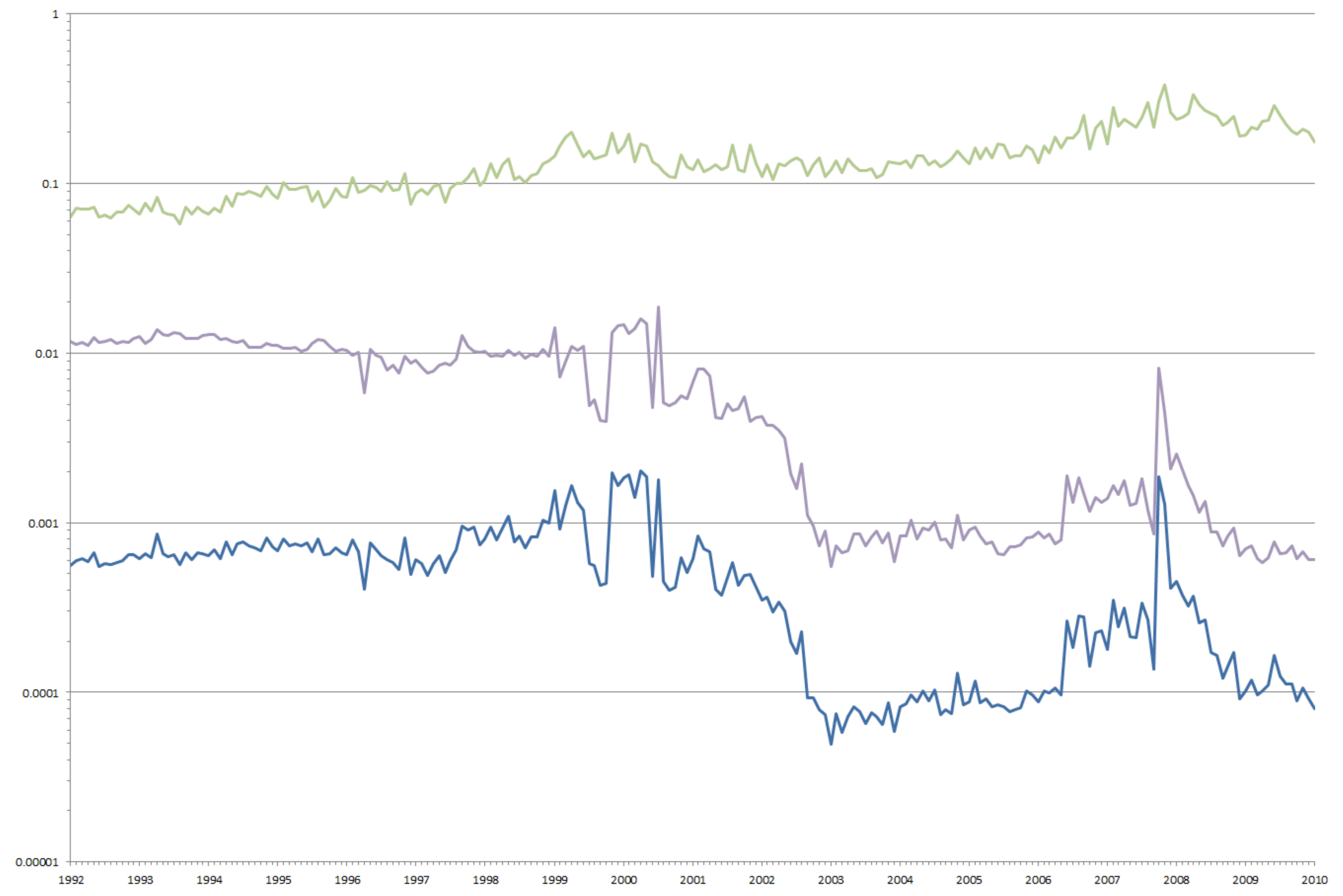}
\begin{tabular}[b]{rrrr}
 & Liquidity  & Share  & Relative\\
Period & Premium & Turnover & Spread\\
\hline\\
1992-1995 & 0.066\% & 7\% & 1.20\%\\
1996-2000 & 0.083\% & 11\% & 0.97\%\\
2001-2005 & 0.038\% & 13\% & 0.37\%\\
2006-2010 & 0.022\% & 21\% & 0.12\%\\
\\
\\
\\
\end{tabular}
\caption{\label{fig:liq_turn_spr}
Left panel: share turnover (top), spread (center), and implied liquidity premium (bottom) in logarithmic scale, from 1992 to 2010. Right panel: monthly averages for share turnover, spread, and implied liquidity premium over subperiods.
Spread and turnover are capitalization-weighted averages across securities in the monthly CRSP database with share codes 10, 11 that have nonzero bid, ask, volume and shares outstanding.}
\end{figure}

The analogies between the comparative statics of the liquidity premium and trading volume suggest a close connection between these quantities. An inspection of the asymptotic formulas unveils the following relations:
\begin{equation}\label{eq:liqest}
{\lip} = \frac34 \eps {\sht} + O(\eps^{4/3})
\qquad\text{and}\qquad
{\left(r+\frac{\mu^2}{2\gamma \sigma^2}\right)-\cer} = \frac34 \eps {\wet} + O(\eps^{4/3}).
\end{equation}
These two relations have the same meaning: the welfare effect of small transaction costs is proportional to trading volume times the spread. The constant of proportionality 3/4 is universal, that is, independent of both  investment opportunities ($r$, $\mu$, $\sigma$) and preferences ($\gamma$).

In the first formula, the welfare effect is measured by the liquidity premium, that is in terms of the risky asset. Likewise, trading volume is expressed as share turnover, which also focuses on the risky asset alone.
By contrast, the second formula considers the decrease in the equivalent safe rate and wealth turnover, two quantities that treat both assets equally.
In summary, if  both welfare and volume are measured consistently with each other, the welfare effect approximately equals volume times the spread, up to the universal factor 3/4.

Figure \ref{fig:liq_turn_spr} plots the spread, share turnover, and the  liquidity premium implied by the first equation in \eqref{eq:liqest}. As in \citet{lo2000trading}, the spread and share turnover are capitalization-weighted averages of all securities in the Center for Research on Security Prices (CRSP) monthly stocks database with share codes 10 and 11, and with nonzero bid, ask, volume and share outstanding. While turnover figures are available before 1992, separate bid and ask prices were not recorded until then, thereby preventing a reliable estimation of spreads for earlier periods. 

Spreads steadily decline in the observation period, dropping by almost an order of magnitude after stock market decimalization of 2001. At the same time, trading volume substantially increases from a typical monthly turnover of 6\% in the early 1990s to over 20\% in the late 2000s. The implied liquidity premium also declines with spreads after decimalization, but less than the spread, in view of the increase in turnover.
During the months of the financial crisis in late 2008, the implied liquidity premium rises sharply, not because of higher volumes, but because spreads widen substantially.
Thus, although this implied liquidity premium is only a coarse estimate, it has advantages over other proxies, because it combines information on both prices and quantities, and is supported by a model.

\subsection{Finite Horizons}

The trading boundaries in this paper are optimal for a long investment horizon, but are also approximately optimal for finite horizons. The following theorem, which complements the main result, makes this point precise:
\begin{theorem}\label{th:finhor}
\label{lem:asympbounds}
Fix a time horizon $T>0$. Then the finite-horizon equivalent safe rate of \emph{any} strategy $(\phi^0,\phi)$ satisfies the upper bound
\begin{align}\label{eq:upperfinite}
\frac{1}{T}\log E\left[(\Xi^\phi_T)^{1-\gamma}\right]^{\frac{1}{1-\gamma}}&\leq r+\frac{\mu^2-\lambda^2}{2\gamma\sigma^2}+\frac{1}{T}\log(\phi^0_{0^-}+\phi_{0^-}S_0)+\pi_*\frac{\epsilon}{T}+O(\varepsilon^{4/3}) ,\\
\intertext{and the finite-horizon equivalent safe rate of our long-run optimal strategy $(\varphi^0,\varphi)$ satisfies the lower bound}
\label{eq:lowerfinite}
\frac{1}{T}\log E\left[(\Xi^\varphi_T)^{1-\gamma}\right]^{\frac{1}{1-\gamma}}&\geq 
r+\frac{\mu^2-\lambda^2}{2\gamma\sigma^2}+\frac{1}{T}\log(\varphi^0_{0^-}+\varphi_{0^-}S_0)-
\left(2\pi_*+\frac{\varphi_{0^-}S_0}{\varphi^0_{0^-}+\varphi_{0^-}S_0}\right)\frac{\varepsilon}{T}+O(\varepsilon^{4/3}).
\end{align}
In particular, for the same unlevered initial position ($\phi_{0^-}=\varphi_{0^-}\ge 0, \phi^0_{0^-}=\varphi^0_{0^-}\ge 0$), the equivalent safe rates of $(\phi^0,\phi)$ and of the optimal policy $(\varphi^0,\varphi)$ for horizon $T$ differ by at most 
\begin{equation}\label{eq:lossbound}
\frac{1}{T}
\left(\log E\left[(\Xi^\phi_T)^{1-\gamma}\right]^{\frac{1}{1-\gamma}}-
\log E\left[(\Xi^\varphi_T)^{1-\gamma}\right]^{\frac{1}{1-\gamma}}\right)\le
(3\pi_*+1)\frac{\varepsilon}T+O(\varepsilon^{4/3}).
\end{equation} 
\end{theorem}

This result implies that the horizon, like consumption, only has a second order effect on portfolio choice with transaction costs, because the finite-horizon equivalent safe rate matches, at the order $\epsilon^{2/3}$, the equivalent safe rate of the stationary long-run optimal policy. This result recovers, in particular, the first-order asymptotics for the finite-horizon value function obtained by \citet[Theorem 4.1]{bichuch.11}. In addition, Theorem \ref{th:finhor} provides explicit estimates for the correction terms of order $\eps$ arising from liquidation costs. Indeed, $r+\frac{\mu^2-\lambda^2}{2\gamma\sigma^2}$ is the maximum rate achieved by trading optimally. The remaining terms arise due to the transient influence of the initial endowment, as well as the costs of the initial transaction, which takes place if the initial position lies outside the no-trade region, and of the final portfolio liquidation. These costs are of order $\varepsilon/T$ because they are incurred only once, and hence defrayed by a longer trading period. By contrast, portfolio rebalancing generates recurring costs, proportional to the horizon, and their impact on the equivalent safe rate does not decline as the horizon increases.

Even after accounting for all such costs in the worst-case scenario, the bound in \eqref{eq:lossbound} shows that their combined effect on the equivalent safe rate is lower than the spread $\varepsilon$, as soon as the horizon exceeds $3\pi_*+1$, that is four years in the absence of leverage. Yet, this bound holds only up to a term of order $\varepsilon^{4/3}$, so it is worth comparing it with the exact bounds in equations \eqref{eq:expl3}-\eqref{eq:expl4}, from which \eqref{eq:upperfinite} and \eqref{eq:lowerfinite} are obtained.

\begin{figure}
\centering
\includegraphics[height=2in]{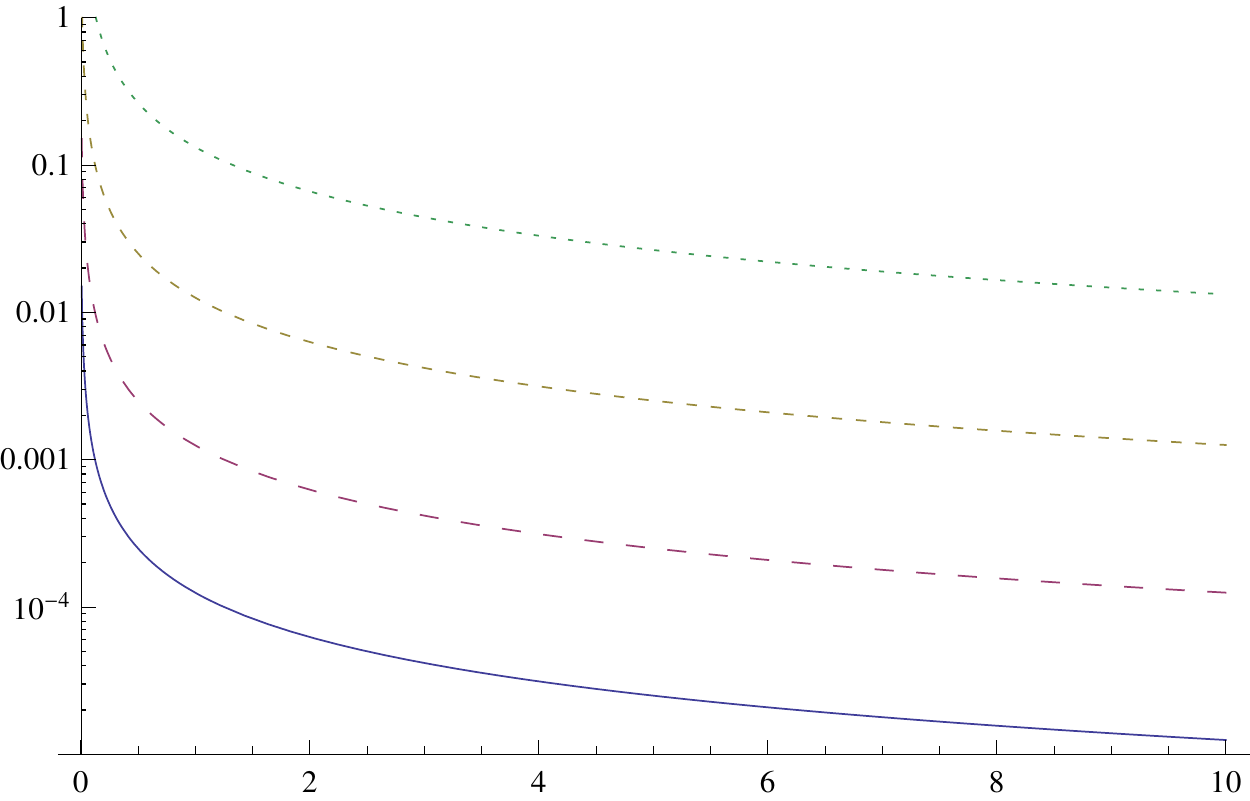}
\caption{\label{fig:finhor}
Upper bound on the difference between the long-run and finite-horizon equivalent safe rates (vertical axis), against the horizon (horizontal axis), for spread $\eps=0.01\%$ (solid), $0.1\%$ (long dashed), $1\%$ (short dashed), and $10\%$ (dotted). Parameters are $\mu=8\%, \sigma=16\%, \gamma=5$.}
\end{figure}

The exact bounds in Figure \ref{fig:finhor} show that, for typical parameter values, the loss in equivalent safe rate of the long-run optimal strategy is lower than the spread $\varepsilon$ even for horizons as short as 18 months, and quickly declines to become ten times smaller, for horizons close to ten years. 
In summary, the long-run approximation is a useful modeling device that makes the model tractable, and the resulting optimal policies are also nearly optimal even for horizons of a few years.

\section{Heuristic Solution}\label{sec:heur}

This section contains an informal derivation of the main results. Here, formal arguments of stochastic control are used to obtain the optimal policy, its welfare, and their asymptotic expansions. 

\subsection{Transaction Costs Market}\label{subsec:tcm}
For a trading strategy $(\varphi^0_t,\varphi_t)$, again write the number of risky shares $\varphi_t=\varphi_t^{\uparrow}-\varphi_t^\downarrow$ as the difference of the cumulated units purchased and sold, and denote by
$$X_t=\varphi^0_t S^0_t, \quad Y_t=\varphi_t S_t,$$
the values of the safe and risky positions in terms of the ask price $S_t$. Then, the self-financing condition~\eqref{eq:selffinancing}, and the dynamics of $S_t^0$ and $S_t$ imply 
\begin{align*}
dX_t =& r X_t dt-S_t d\varphi^{\uparrow}_t+(1-\ve)S_t d\varphi^{\downarrow}_t,\\
dY_t =& (\mu+r) Y_t dt +\sigma Y_t dW_t +S_td\varphi^{\uparrow}_t-S_td\varphi^{\downarrow}.
\end{align*}
Consider the maximization of expected power utility $U(x)=x^{1-\gamma}/ (1-\gamma) $ from terminal wealth at time $T$,\footnote{For a fixed horizon $T$, one would need to specify whether terminal wealth is valued at bid, ask, or at liquidation prices, as in Definition \ref{def:admiss}. In fact, since these prices are within a constant positive multiple of each other, which price is used is inconsequential for a long-run objective. For the same reason, the terminal condition for the finite horizon value function does not have to be satisfied by the stationary value function, because its effect is negligible.} and denote by $V(t,x,y)$ its value function, which depends on time and the value of the safe and risky positions. It\^o's formula yields:
\begin{align*}
d V(t,X_t,Y_t)=& V_t dt+V_x dX_t + V_y dY_t +\frac 12  V_{yy} d\langle Y,Y\rangle_t\\
=& 
\left(V_t+ r X_t V_x+(\mu+r)Y_t V_y+\frac{\sigma^2}2 Y_t^2 V_{yy}\right)dt\\
&+ S_t(V_y-V_x)d\varphi^{\uparrow}_t+S_t((1-\ve)V_x-V_y)d\varphi^{\downarrow}_t+\sigma Y_t V_y dW_t,
\end{align*}
where the arguments of the functions are omitted for brevity. By the martingale optimality principle of stochastic control (cf.\ \citet*{fleming2006controlled}), the value function $V(t,X_t,Y_t)$ must be a supermartingale for any choice of the cumulative purchases and sales $\varphi_t^{\uparrow},\varphi_t^{\downarrow}$. Since these are increasing processes, it follows that  $V_y-V_x \le 0$ and $(1-\ve) V_x-V_y \le 0$, that is
\begin{equation*}
1 \le \frac{V_x}{V_y}\le \frac{1}{1-\ve}.
\end{equation*}
In the interior of this ``no-trade region'', where the number of risky shares remains constant, the drift of $V(t,X_t,Y_t)$ cannot be positive, and must become zero for the optimal policy:\footnote{Alternatively, this equation can be obtained from standard arguments of singular control, cf.\ \citet*[Chapter VIII]{fleming2006controlled}.}
\begin{equation}\label{eq:hjb}
 V_t+ r X_t V_x+(\mu+r)Y_t V_y+\frac{\sigma^2}2 Y_t^2 V_{yy} =0
 \qquad \text{if } \qquad 1< \frac{V_x}{V_y}<\frac{1}{1-\ve}.
\end{equation}
To simplify further, note that the value function must be homogeneous with respect to wealth, and that --- in the long run --- it should grow exponentially with the horizon at a constant rate. 
These arguments lead to guess\footnote{This guess assumes that the cash position is strictly positive, $X_t>0$, which excludes leverage. With leverage, factoring out $(-X_t)^{1-\gamma}$ leads to analogous calculations. In either case, under the optimal policy, the ratio $Y_t/X_t$ always remains either strictly positive, or strictly negative, never to pass through zero.} that $V(t,X_t,Y_t)=(X_t)^{1-\gamma}v(Y_t/X_t) e^{- (1-\gamma)  (r+\beta) t}$ for some $\beta$ to be found. Setting $z=y/x$, the above equation reduces to
\begin{equation}\label{eq:hjbred}
\frac{\sigma^2}2 z^2v''(z)+\mu z v'(z)- (1-\gamma) \beta v(z)=0
\qquad \text{if } \qquad 1+z<\frac{ (1-\gamma) v(z)}{v'(z)}<\frac{1}{1-\ve}+z.
\end{equation}
Assuming that the no-trade region $\{z:1+z\leq \frac{ (1-\gamma) v(z)}{v'(z)}\leq\frac{1}{1-\ve}+z\}$ coincides with some interval $l\le z\le u$ to be determined, and noting that at $l$ the left inequality in \eqref{eq:hjbred} holds as equality, while at $u$ the right inequality holds as equality, the following free boundary problem arises:
\begin{align}
\label{hjbred}
\frac{\sigma^2}2 z^2v''(z)+\mu z v'(z)- (1-\gamma) \beta v(z)&=0
\qquad \text{if } l < z < u, \\
\label{boundbuy}
(1+l)v'(l)-(1-\gamma)v(l)&=0,\\
\label{boundsell}
(1/(1-\ve)+u)v'(u)-(1-\gamma)v(u)&=0.
\end{align}
These conditions are not enough to identify the solution, because they can be matched for any choice of the trading boundaries $l, u$. The optimal boundaries are the ones that also satisfy the smooth-pasting conditions (cf. \citet*{benes.al.80,dumas.91}), formally obtained by differentiating \eqref{boundbuy} and \eqref{boundsell} with respect to $l$ and $u$, respectively:
\begin{align}
\label{smoothbuy}
(1+l)v''(l)+\gamma v'(l)=0,\\
(1/(1-\ve)+u)v''(u)+ \gamma v'(u)=0.\label{smoothsell}
\end{align}
In addition to the reduced value function $v$, this system requires to solve for the excess equivalent safe rate $\beta$ and the trading boundaries $l$ and $u$. Substituting~\eqref{smoothbuy} and~\eqref{boundbuy}  into~\eqref{hjbred} yields (cf.\  \citet*{dumas.luciano.91})
\begin{align*}
-\frac{\sigma^2}2 (1-\gamma)\gamma\frac{l^2}{(1+l)^2}v +\mu (1-\gamma)\frac{l}{1+l}v -(1-\gamma)\beta v=0.
\end{align*}
Setting $\pi_-=l/(1+l)$, and factoring out $(1-\gamma) v$, it follows that 
\begin{align*}
-\frac{\gamma \sigma^2}2 \pi_-^2+\mu \pi_- -\beta=0.
\end{align*}
Note that $\pi_-$ is the risky weight when it is time to buy, and hence the risky position is valued at the ask price. 
The same argument for $u$ shows that the other solution of the quadratic equation is $\pi_+=u (1-\ve)/(1+u (1-\ve))$, which is the risky weight when it is time to sell, and hence the risky position is valued at the bid price. Thus, the optimal policy is to buy when the ``ask" fraction falls below~$\pi_-$, sell when the ``bid" fraction rises above~$\pi_+$, and do nothing in between.
Since $\pi_-$ and $\pi_+$ solve the same quadratic equation, they are related to $\beta$ via
\begin{align*}
\pi_\pm=\frac{\mu}{\gamma \sigma^2}\pm \frac{\sqrt{\mu^2-2 \beta\gamma\sigma^2}}{\gamma\sigma^2}.
\end{align*}
It is convenient to set $\beta=(\mu^2-\lambda^2)/2\gamma\sigma^2$, because $\beta=\mu^2/2\gamma\sigma^2$ without transaction costs. We call $\lambda$ the \emph{gap}, since $\lambda=0$ in a frictionless market, and, as $\lambda$ increases, all variables diverge from their frictionless values.
Put differently, to compensate for transaction costs, the investor would require another asset, with expected return $\lambda$ and volatility $\sigma$, which trades without frictions and is uncorrelated with the risky asset.\footnote{Recall that in a frictionless market with two uncorrelated assets with returns $\mu_1$ and $\mu_2$, both with volatility $\sigma$, the maximum Sharpe ratio is $(\mu_1^2+\mu_2^2)/\sigma^2$. That is, squared Sharpe ratios add across orthogonal shocks.}
With this notation, the buy and sell boundaries are just
\begin{align*}\label{eq:uplow}
\pi_\pm=\frac{\mu\pm\lambda}{\gamma \sigma^2}.
\end{align*}
In other words, the buy and sell boundaries are symmetric around the classical frictionless solution $\mu/\gamma \sigma^2$. 
Since $l(\lambda),u(\lambda)$ are identified by $\pi_\pm$ in terms of $\lambda$, it now remains to find $\lambda$. After deriving $l(\lambda)$ and $u(\lambda)$, the boundaries in the problem \eqref{hjbred}-\eqref{boundsell} are no longer free, but fixed.
\nada{
\begin{align*}
\frac{\sigma^2}2 z^2v''(z)+\mu z v'(z)+(1-1/\gamma)\frac{\mu^2-\lambda^2}{2\sigma^2}v(z) &= 0
\qquad \text{if } z\in [l(\lambda),u(\lambda)], \\
((1-\ve)+l(\lambda))v'(l(\lambda))-(1-\gamma)v(l(\lambda)) &= 0,\\
(1/(1-\ve)+u(\lambda))v'(u(\lambda))-(1-\gamma)v(u(\lambda)) &= 0.
\end{align*}}
With the substitution
\[
  v(z)=e^{(1-\gamma)\int_0^{\log(z/l(\lambda))} w(y)dy}, \quad \mbox{i.e.,} \quad w(y)=\frac{l(\lambda)e^y v'(l(\lambda)e^y)}{(1-\gamma)v(l(\lambda)e^y)}, \]
the boundary problem \eqref{hjbred}-\eqref{boundsell} reduces to a Riccati ODE
\begin{align}
w'(y)+(1-\gamma)w(y)^2+\left(\frac{2\mu}{\sigma^2}-1\right)w(y)-
\gamma \left(\frac{\mu-\lambda}{\gamma \sigma^2}\right)
\left(\frac{\mu+\lambda}{\gamma \sigma^2}\right)
 &= 0, \quad y\in[0,\log u(\lambda)/l(\lambda)], \label{riccati1}\\
w(0) &= \frac{\mu-\lambda}{\gamma \sigma^2},\label{riccati2}\\
w(\log (u(\lambda)/l(\lambda))) &= \frac{\mu+\lambda}{\gamma \sigma^2},\label{riccati3}
\end{align}
where
\begin{equation}
\frac{u(\lambda)}{l(\lambda)} = 
\frac{1}{(1-\ve)}\frac{\pi_+ (1-\pi_-)}{\pi_- (1-\pi_+)} =
\frac{1}{(1-\ve)}\frac{(\mu+\lambda)(\mu-\lambda-\gamma \sigma^2)}
{(\mu-\lambda)(\mu+\lambda-\gamma \sigma^2)}.
\end{equation}
For each $\lambda$, the initial value problem \eqref{riccati1}-\eqref{riccati2} has a solution $w(\lambda,\cdot)$, and the correct value of $\lambda$ is identified by the second boundary condition~\eqref{riccati3}.

\subsection{Asymptotics}\label{subsec:asymp}
The equation \eqref{riccati3} does not have an explicit solution, but it is possible to obtain an asymptotic expansion for small transaction costs ($\eps\sim 0$) using the implicit function theorem. To this end, write the boundary condition~\eqref{riccati3} as $f(\lambda,\eps)=0$, where:
\begin{equation*}\label{eq:flameps}
f(\lambda,\eps) = w(\lambda, \log(u(\lambda)/l(\lambda)))-\frac{\mu+\lambda}{\gamma \sigma^2}.
\end{equation*}
Of course, $f(0,0)=0$ corresponds to the frictionless case. The implicit function theorem then suggests that around zero $\lambda(\eps)$ follows the asymptotics $\lambda(\eps) \sim - \eps f_\eps/f_\lambda$, but the difficulty is that $f_\lambda=0$, because $\lambda$ is not of order $\eps$. Heuristic arguments \citep*{MR1284980,rogers.01} suggest that $\lambda$ is of order $\eps^{1/3}.$\footnote{Since $\lambda$ is proportional to the width of the no-trade region $\delta$, the question is why the latter is of order $\eps^{1/3}$. The intuition is that a no-trade region of width $\delta$ around the frictionless optimum leads to transaction costs of order $\eps/\delta$ (because the time spent near the boundaries is approximately inversely proportional to the length of the interval), and to a welfare cost of the order $\delta^2$ (because the region is centered around the frictionless optimum, hence the linear welfare cost is zero). Hence, the total cost is of the order  $\eps/\delta + \delta^2$, and attains its minimum for $\delta = O(\eps^{1/3})$. } Thus, setting $\lambda=\delta^{1/3}$ and $\hat f(\delta,\eps)=f(\delta^{1/3},\eps)$, and computing the derivatives of the explicit formula for $w(\lambda,x)$ (cf.\ Lemma \ref{lem:riccati}) shows that: 
\begin{align*}
\hat f_\eps(0,0) = -\frac{\mu  \left(\mu -\gamma  \sigma ^2\right)}{\gamma ^2 \sigma ^4}, \qquad \hat f_\delta(0,0) = \frac{4}{3 \mu ^2 \sigma ^2-3 \gamma  \mu  \sigma ^4}.
\end{align*}
As a result:
\begin{equation*}
\delta(\eps) \sim -\frac{f_\eps}{f_\delta} \eps = 
\frac{3 \mu ^2 \left(\mu -\gamma  \sigma^2\right)^2}{4 \gamma ^2 \sigma ^2}\eps
\quad\text{whence}\quad
\lambda(\eps) \sim
\left(\frac{3\mu ^2 \left(\mu -\gamma  \sigma ^2\right)^2}{4\gamma ^2 \sigma ^2}\right)^{1/3} \eps^{1/3}.
\end{equation*}
The asymptotic expansions of all other quantities then follow by Taylor expansion.

\section{Conclusion} 

In a tractable model of transaction costs with one safe and one risky asset and constant investment opportunities, we have computed explicitly the optimal trading policy, its welfare, liquidity premium, and trading volume, for an investor with constant relative risk aversion and a long horizon.

The trading boundaries are symmetric around the Merton proportion, if each boundary is computed with the corresponding trading price. Both the liquidity premium and trading volume are small in the unlevered regime, but become substantial in the presence of leverage. For a small bid-ask spread, the liquidity premium is approximately equal to share turnover times the spread, times the universal constant 3/4.

Trading boundaries depend on investment opportunities only through the mean variance ratio. The equivalent safe rate, the liquidity premium, and trading volume also depend only on the mean variance ratio if measured in business time.

\appendix
\section*{Appendix}\label{sec:verification}

\section{Explicit Formulas and their Properties}

We now show that the candidate $w$ for the reduced value function and the quantity $\lambda$ are indeed well-defined for sufficiently small spreads.  The first step is to determine, for a given small $\lambda>0$, an explicit expression for the solution $w$ of the ODE~\eqref{riccati1}, complemented by the initial condition~\eqref{riccati2}.

\begin{lemma}\label{lem:riccati}
Let $0<\mu/\gamma \sigma^2 \neq 1$.
Then for sufficiently small $\lambda>0$, the function
$$
w(\lambda,y)=\begin{cases} 
\frac{a(\lambda)\tanh[\tanh^{-1}(b(\lambda)/a(\lambda))-a(\lambda)y]+(\frac{\mu}{\sigma^2}-\frac{1}{2})}{\gamma-1}, &\mbox{if } \gamma \in (0,1) \mbox{ and } \frac{\mu}{\gamma\sigma^2}<1 \mbox{ or } \gamma>1 \mbox{ and } \frac{\mu}{\gamma\sigma^2}>1,\\
\frac{a(\lambda) \tan[\tan^{-1}(b(\lambda)/a(\lambda))+a(\lambda)y]+(\frac{\mu}{\sigma^2}-\frac{1}{2})}{\gamma-1}, &\mbox{if } \gamma>1 \mbox{ and } \frac{\mu}{\gamma \sigma^2} \in \left(\frac{1}{2}-\frac{1}{2}\sqrt{1-\frac{1}{\gamma}},\frac{1}{2}+\frac{1}{2}\sqrt{1-\frac{1}{\gamma}}\right),\\
\frac{a(\lambda)\coth[\coth^{-1}(b(\lambda)/a(\lambda))-a(\lambda)y]+(\frac{\mu}{\sigma^2}-\frac{1}{2})}{\gamma-1}, &\mbox{otherwise},
\end{cases}
$$  
with
\[
a(\lambda)=\sqrt{\Big|(\gamma-1)\frac{\mu^2-\lambda^2}{\gamma\sigma^4}-\Big(\frac{1}{2}-\frac{\mu}{\sigma^2}\Big)^2\Big|} \quad \mbox{and} \quad b(\lambda)=\frac{1}{2}-\frac{\mu}{\sigma^2}+(\gamma-1)\frac{\mu-\lambda}{\gamma\sigma^2},
\]
is a local solution of 
\begin{equation}\label{eq:wode}
w'(y)+(1-\gamma)w^2(y)+\left(\frac{2\mu}{\sigma^2}-1\right)w(y)-\frac{\mu^2-\lambda^2}{\gamma \sigma^4}=0, \quad w(0)=\frac{\mu-\lambda}{\gamma\sigma^2}.
\end{equation}
Moreover, $y \mapsto w(\lambda,x)$ is increasing (resp.\ decreasing) for $\mu/\gamma\sigma^2 \in (0,1)$ (resp.\ $\mu/\gamma\sigma^2>1$).
\end{lemma}

\begin{proof}
The first part of the assertion is easily verified by taking derivatives, noticing that the case distinctions distinguish between the different signs of the discriminant
$$(\gamma-1)\frac{\mu^2-\lambda^2}{\gamma\sigma^4}-\left(\frac{1}{2}-\frac{\mu}{\sigma^2}\right)^2$$
of the Riccati equation \eqref{eq:wode} for sufficiently small $\lambda$. Indeed, in the second case the discriminant is positive for sufficiently small $\lambda$. The first and third case correspond to a negative discriminant, as well as $b(\lambda)/a(\lambda)<1$ and $b(\lambda)/a(\lambda)>1$, respectively, for sufficiently small $\lambda>0$, so that the function $w$ is well-defined in each case.

The second part of the assertion follows by inspection of the explicit formulas.
\end{proof}

Next, establish that the crucial constant $\lambda$, which determines both the no-trade region and the equivalent safe rate, is well-defined.

\begin{lemma}\label{lem:lambda}
Let $0<\mu/\gamma \sigma^2 \neq 1$ and $w(\lambda,\cdot)$ be defined as in
Lemma~\ref{lem:riccati}, and set
$$
l(\lambda)=\frac{\mu-\lambda}{\gamma\sigma^2-(\mu-\lambda)}, \quad u(\lambda)=\frac{1}{(1-\ve)}\frac{\mu+\lambda}{\gamma\sigma^2-(\mu+\lambda)}.
$$
Then, for sufficiently small $\ve>0$, there exists a unique solution $\lambda$ of 
\begin{equation}\label{eq:wrbd}
  w\left(\lambda,\log\left(\frac{u(\lambda)}{l(\lambda)}\right)\right)
  -\frac{\mu+\lambda}{\gamma\sigma^2}=0.
\end{equation}
As $\ve \downarrow 0$, it has the asymptotics
\[
  \lambda = \gamma\sigma^2\left(\frac{3}{4\gamma}\left(\frac{\mu}{\gamma\sigma^2}\right)^2\left(1-\frac{\mu}{\gamma\sigma^2}\right)^2\right)^{1/3}\ve^{1/3}+\sigma^2\left(\frac{(5-2\gamma)}{10}\frac{\mu}{\gamma\sigma^2}\left(1-\frac{\mu}{\gamma\sigma^2}\right)-\frac{3}{20}\right)\ve+O(\ve^{4/3}).
\]
\end{lemma}
\begin{proof}
The explicit expression for~$w$ in Lemma~\ref{lem:riccati} 
implies that $w(\lambda,x)$ in Lemma~\ref{lem:riccati} is analytic in both
  variables at $(0,0)$. By the initial condition in~\eqref{eq:wode}, its power series has the form
  \[
    w(\lambda,x) = \frac{\mu-\lambda}{\gamma\sigma^2}
      +\sum_{i=1}^\infty \sum_{j=0}^\infty W_{ij} x^i \lambda^j,
  \]
  where expressions for the coefficients~$W_{ij}$ are computed
  by expanding the explicit expression for~$w$. (The leading terms are provided after this proof.)
  Hence, the left-hand side of the boundary condition~\eqref{eq:wrbd}
  is an analytic function of~$\ve$ and~$\lambda$. Its power series
  expansion shows that the coefficients of $\ve^0\lambda^{j}$ vanish for
  $j=0,1,2$, so that the condition~\eqref{eq:wrbd} reduces to
  \begin{equation}\label{eq:la eq}
    \lambda^3 \sum_{i\geq0} A_i \lambda^i = \ve \sum_{i,j\geq0}
      B_{ij} \ve^i \lambda^j
  \end{equation}
  with (computable) coefficients~$A_i$ and~$B_{ij}$. This equation has to be solved
  for~$\lambda$.
  Since
  \[
    A_0 = \frac{4}{3\mu\sigma^2(\gamma\sigma^2-\mu)} \qquad \text{and}
      \qquad B_{00} = \frac{\mu(\gamma\sigma^2-\mu)}{\gamma^2\sigma^4}
  \]
  are non-zero, divide the equation~\eqref{eq:la eq} by
  $\sum_{i\geq0} A_i \lambda^i$, and take the third root, obtaining that, for some~$C_{ij}$,
  \[
    \lambda = \ve^{1/3} \sum_{i,j\geq0} C_{ij} \ve^i \lambda^j
      = \ve^{1/3} \sum_{i,j\geq0} C_{ij} (\ve^{1/3})^{3i} \lambda^j \, .
  \]
  The right-hand side is an analytic function
  of~$\lambda$ and~$\ve^{1/3}$, so that the implicit function
  theorem~\citep*[Theorem~I.B.4]{MR2568219}
  yields a unique solution~$\lambda$ (for~$\ve$ sufficiently small),
  which is an analytic function of~$\ve^{1/3}$.
  Its power series coefficients can be computed at any order. 
  \end{proof}
  
In the preceding proof we needed the first coefficients of the series expansion of the analytic function on the left-hand side of~\eqref{eq:wrbd}.
Calculating them is elementary, but rather cumbersome, and can be quickly performed with symbolic computation software.
Following a referee's suggestion, we present some expressions to aid readers who wish to check the calculations by hand, namely
the derivatives of $w$ at $(\lambda,x)=(0,0)$ that are needed to calculate the Taylor coefficients of~\eqref{eq:wrbd} used in the proof.
Note that they are the same in all three cases of Lemma~\ref{lem:riccati}:

\begin{align*}
  w_x(0,0) &= -\frac{\mu^2}{\gamma^2\sigma^4}+\frac{\mu}{\gamma\sigma^2}, \quad w_\lambda(0,0) = -\frac{1}{\gamma\sigma^2},\\
  w_{xx}(0,0) &= \frac{2\mu^3}{\gamma^3 \sigma^6}-\frac{3\mu^2}{\gamma^2 \sigma^4}+\frac{\mu}{\gamma \sigma^2}, \quad w_{x\lambda}(0,0) = \frac{2\mu}{\gamma^2 \sigma^4}-\frac{1}{\gamma \sigma^2}, \quad w_{\lambda\lambda}(0,0) = 0, \\
   w_{xxx}(0,0) &= -\frac{6\mu^4}{\gamma^4 \sigma^8} + \frac{2\mu^4}{\gamma^3 \sigma^8}
    +\frac{12\mu^3}{\gamma^3 \sigma^6} - \frac{4\mu^3}{\gamma^2 \sigma^6}
    -\frac{7\mu^2}{\gamma^2 \sigma^4} + \frac{2\mu^2}{\gamma \sigma^4}
    +\frac{\mu}{\gamma \sigma^2}, \\
  w_{xx\lambda}(0,0) &= -\frac{6\mu^2}{\gamma^3 \sigma^6} + \frac{2\mu^2}{\gamma^2 \sigma^6}
    +\frac{6\mu}{\gamma^2 \sigma^4} -\frac{2\mu}{\gamma \sigma^4}
    - \frac{1}{\gamma \sigma^2}, \quad w_{x\lambda\lambda}(0,0) = -\frac{2}{\gamma^2\sigma^4}, \quad w_{\lambda\lambda\lambda}(0,0) = 0. 
\end{align*}

Henceforth, consider small transaction costs $\ve>0$,
and let~$\lambda$ denote the constant in Lemma~\ref{lem:lambda}. Moreover, set $w(y)=w(\lambda,y)$, $a=a(\lambda)$, $b=b(\lambda)$, and $u=u(\lambda)$, $l=l(\lambda)$. In all cases, the function $w$ can be extended smoothly to an open neighborhood of $[0,\log(u/l)]$ (resp.\ $[\log(u/l),0]$ if $\mu/\gamma\sigma^2>1$). By continuity, the ODE \eqref{eq:wode} then also holds at $0$ and $\log(u/l)$; inserting the boundary conditions for $w$ in turn readily yields the following counterparts for the derivative $w'$:

\begin{lemma}\label{lem:smoothpasting}
Let  $0<\mu/\gamma \sigma^2 \neq 1$. Then, in all three cases,
$$w'(0)=\frac{\mu-\lambda}{\gamma\sigma^2}-\left(\frac{\mu-\lambda}{\gamma\sigma^2}\right)^2, \quad w'\left(\log\left(\frac{u}{l}\right)\right)=\frac{\mu+\lambda}{\gamma\sigma^2}-\left(\frac{\mu+\lambda}{\gamma\sigma^2}\right)^2.$$
\end{lemma}

\section{Shadow Prices and Verification}

The key to justify the heuristic arguments of Section \ref{sec:heur} is to reduce the portfolio choice problem with transaction costs to another portfolio choice problem, without transaction costs. Here, the bid and ask prices are replaced by a single {shadow} price $\tilde{S}_t$, evolving within the bid-ask spread, which coincides with either price at times of trading, and yields the same optimal policy and utility. Evidently, \emph{any} frictionless market extension with values in the bid-ask spread leads to more favorable terms of trade than the original market with transaction costs. To achieve equality, the particularly unfavorable shadow price must match the trading prices whenever its optimal policy transacts.

\begin{definition}\label{defi:shadow}
A \emph{shadow price} is a frictionless price process $\tilde{S}_t$, evolving within the bid-ask spread ($(1-\ve)S_t \le \tilde S_t \le S_t$ a.s.), such that there is an optimal strategy for $\tilde{S}_t$ which is of finite variation, and entails buying only when the shadow price $\tilde{S}_t$ equals the ask price $S_t$, and selling only when $\tilde{S}_t$ equals the bid price $(1-\ve)S_t$.
\end{definition}

Once a candidate for such a shadow price is identified, long-run verification results for frictionless models (cf.~\citet*{guasoni.robertson.10}) deliver the optimality of the guessed policy. Further, this method provides explicit upper and lower bounds on finite-horizon performance (cf.\ Lemma \ref{lemfinite} below), thereby allowing to check whether the long-run optimal strategy is approximately optimal for an horizon $T$. Put differently, it shows which horizons are long enough.

\subsection{Derivation of a Candidate Shadow Price}\label{subsec:sm}

With a smooth candidate value function at hand, a candidate shadow price can be identified as follows. By definition, trading the shadow price should not allow the investor to outperform the original market with transaction costs. In particular, if $\tilde{S}_t$ is the value of the shadow price at time $t$, then allowing the investor to carry out at single trade at time $t$ at this \emph{frictionless} price should not lead to an increase in utility. A trade of $\nu$ risky shares at the frictionless price $\tilde{S}_t$ moves the investor's safe position $X_t$ to $X_t-\nu \tilde{S}_t$ and her risky position (valued at the ask price $S_t$) from $Y_t$ to $Y_t+\nu S_t$. Then -- recalling that the second and third arguments of the candidate value function $V$ from the Section \ref{sec:heur} were precisely the investor's safe and risky positions -- the requirement that such a trade does not increase the investor's utility is tantamount to:
$$V(t,X_t-\nu \tilde{S}_t,Y_t+\nu S_t) \leq V(t,X_t,Y_t), \quad \forall \nu \in \mathbb{R}.$$
A Taylor expansion of the left-hand side for small $\nu$ then implies that $-\nu \tilde{S}_t V_x+\nu S_t V_y \leq 0$. Since this inequality has to hold both for positive and negative values of $\nu$, it yields
\begin{equation}\label{eq:mrs}
\tilde{S}_t=\frac{V_y}{V_x} S_t.
\end{equation}
That is, the multiplicative deviation of the  shadow price from the ask price should be the marginal rate of substitution of risky for safe assets. In particular, this argument immediately yields a candidate shadow price, once a smooth candidate value function has been identified. For the long-run problem, we derived the following candidate value function in the previous section: 
$$V(t,X_t,Y_t)=e^{-(1-\gamma)(r+\beta)t}(X_t)^{1-\gamma} e^{(1-\gamma)\int_0^{\log(Y_t/lX_t)}w(y)dy}.$$
Using this equality to calculate the partial derivatives in \eqref{eq:mrs}, the candidate shadow price becomes:
\begin{equation}\label{eq:defshadow}
\tilde{S}_t=\frac{w(\Upsilon_t)}{le^{\Upsilon_t}(1-w(\Upsilon_t))}S_t,
\end{equation}
where $\Upsilon_t=\log(Y_t/lX_t)$ denotes the logarithm of the risky-safe ratio, centered at its value at the lower buying boundary $l$. If this candidate is indeed the right one, then its optimal strategy and value function should coincide with their frictional counterparts derived heuristically above. In particular, the optimal risky fraction $\tilde{\pi}_t$ should correspond to the same numbers $\varphi^0_t$ and $\varphi_t$ of safe and risky shares, if measured in terms of $\tilde{S}_t$ instead of the ask price $S_t$. As a consequence:
\begin{equation}\label{eq:tildew}
\tilde{\pi}_t=\frac{\varphi_t \tilde{S}_t}{\varphi^0_tS^0_t+\varphi_t\tilde{S}_t}=\frac{\varphi_t S_t \frac{w(\Upsilon_t)}{le^{\Upsilon_t}(1-w(\Upsilon_t))}}{\varphi^0_t S^0_t +\varphi_t S_t \frac{w(\Upsilon_t)}{le^{\Upsilon_t}(1-w(\Upsilon_t))}}=\frac{\frac{w(\Upsilon_t)}{1-w(\Upsilon_t)}}{1+\frac{w(\Upsilon_t)}{1-w(\Upsilon_t)}}=w(\Upsilon_t),
\end{equation}
where, for the third equality, we have used that the risky-safe ratio $\varphi_t S_t/\varphi^0_t S^0_t$ can be written as $l e^{\Upsilon_t}$ by the definition of $\Upsilon_t$. We now turn to the corresponding frictionless value function $\tilde{V}$. By the definition of a shadow price, it should coincide with its frictional counterpart $V$. In the frictionless case, it is more convenient to factor out the total wealth $\tilde{X}_t=\varphi^0_t S^0_t+\varphi_t \tilde{S}_t$ (in terms of the frictionless risky price $\tilde{S}_t$) instead of the safe position $X_t=\varphi^0_t S^0_t$, giving
$$\tilde{V}(t,\tilde{X}_t,\Upsilon_t)=V(t,X_t,Y_t)=e^{-(1-\gamma)(r+\beta)t} \tilde{X}_t^{1-\gamma} \left(\frac{X_t}{\tilde{X}_t}\right)^{1-\gamma} e^{(1-\gamma)\int_0^{\Upsilon_t}w(y)dy}.$$
Since $X_t/\tilde{X}_t=1-w(\Upsilon_t)$ by definition of $\tilde{S}_t$ and $\Upsilon_t$, one can rewrite the last two factors as
\begin{align*}
\left(\frac{X_t}{\tilde{X}_t}\right)^{1-\gamma} e^{(1-\gamma)\int_0^{\Upsilon_t}w(y)dy}&=\exp\left((1-\gamma)\left[\log(1-w(\Upsilon_t))+\int_0^{\Upsilon_t} w(y)dy\right]\right)\\
&=(1-w(0))^{\gamma-1} \exp\left((1-\gamma)\int_0^{\Upsilon_t} \left(w(y)-\frac{w'(y)}{1-w(y)}\right) dy\right).
\end{align*}
Then, setting $\tilde{w}=w-\frac{w'}{1-w}$, the candidate long-run value function for $\tilde{S}$ becomes
$$\tilde{V}(t,\tilde{X}_t,\Upsilon_t)=e^{-(1-\gamma)(r+\beta)t} \tilde{X}_t^{1-\gamma} e^{(1-\gamma)\int_0^{\Upsilon_t}\tilde{w}(y)dy}(1-w(0))^{\gamma-1}.$$
Starting from the candidate value function and optimal policy for $\tilde{S}$, we can now proceed to verify that they are indeed optimal for $\tilde{S}_t$, by adapting the argument from \cite{guasoni.robertson.10}. But before we do that, we have to construct the respective processes.

\subsection{Construction of the Shadow Price}

The above heuristic arguments suggest that the optimal ratio $Y_t/X_t=\varphi_t S_t/\varphi^0_t S^0_t$ should take values in the interval $[l,u]$. Hence, $\Upsilon_t=\log(Y_t/lX_t)$ should be $[0,\log(u/l)]$-valued if the lower trading boundary $l$ for the ratio $Y_t/X_t$ is positive. If the investor shorts the safe asset to leverage her risky position, the ratio becomes negative. In the frictionless case, and also for small transaction costs, this happens if risky weight $\mu/\gamma\sigma^2$ is bigger than $1$. Then, the trading boundaries $l \leq u$ are both negative, so that the centered log-ratio $\Upsilon_t$ should take values in $[\log(u/l),0]$. In both cases, trading should only take place when the risky-safe ratio reaches the boundaries of this region. Hence, the numbers of safe and risky units $\varphi^0_t$ and $\varphi_t$ should remain constant and $\Upsilon_t=\log(\varphi_t/l\varphi^0_t)+\log(S_t/S^0_t)$ should follow a Brownian motion with drift as long as $\Upsilon_t$ moves in $(0,\log(u/l))$ (resp.\ in $(\log(u/l),0)$ if $\mu/\gamma\sigma^2>1$). This argument motivates the \emph{definition} of the process $\Upsilon_t$ as reflected Brownian motion:
\begin{equation}\label{eq:reflected}
d\Upsilon_t=(\mu-\sigma^2/2)dt+\sigma dW_t+dL_t-dU_t, \quad \Upsilon_0 \in [0,\log(u/l)],
\end{equation}
for continuous, adapted local time processes $L$ and $U$ which are nondecreasing (resp.\ nonincreasing if $\mu/\gamma\sigma^2>1$) and increase (resp.\ decrease if $\mu/\gamma\sigma^2>1$) only on the sets $\{\Upsilon_t=0\}$ and $\{\Upsilon_t=\log(u/l)\}$, respectively. Starting from this process, the existence of which is a classical result of \cite{skorokhod.61}, the process $\tilde{S}$ is defined in accordance with \eqref{eq:defshadow}:

\begin{lemma}\label{lem:dynamics}
Define
\begin{equation}\label{eq:jump}
\Upsilon_0=\begin{cases} 0, &\mbox{if } l\xi^0S^0_0 \geq \xi S_0,\\  \log(u/l), &\mbox{if } u\xi^0 S^0_0 \leq \xi S_0,\\ \log[(\xi S_0/\xi^0 S^0_0)/l], &\mbox{otherwise,}  \end{cases}
\end{equation}
and let $\Upsilon$ be defined as in~\eqref{eq:reflected}, starting at $\Upsilon_0$. Then, $\tilde S = S \frac{w(\Upsilon)}{l e^{\Upsilon} (1-w(\Upsilon))}$, with $w$ as in Lemma~\ref{lem:riccati}, has the dynamics
\begin{displaymath}
d\tilde S_t/\tilde S_t= \left(\tilde{\mu}(\Upsilon_t)+r\right)d t+ \tilde{\sigma}(\Upsilon_t)d W_t,
\end{displaymath}
where $\tilde \mu(\cdot)$ and $\tilde \sigma (\cdot)$ are defined as
\begin{displaymath}
\tilde{\mu}(y) = \frac{\sigma^2w'(y)}{w(y)(1-w(y))}\left(\frac{w'(y)}{1-w(y)}-(1-\gamma)w(y)\right), \quad \tilde{\sigma}(y) = \frac{\sigma w'\left(y\right)}{w(y)(1-w(y))}.
\end{displaymath}
Moreover, the process $\tilde S$ takes values within the bid-ask spread $[(1-\varepsilon)S,S]$.
\end{lemma}
Note that the first two cases in~\eqref{eq:jump} arise if the initial risky-safe ratio $\xi S_0/(\xi^0 S_0^0)$ lies outside of the interval $[l,u]$. Then, a jump from the initial position $(\varphi_{0^-}^0, \varphi_{0^-}) = (\xi^0,\xi)$ to the nearest boundary value of $[l,u]$ is required. This transfer requires the purchase resp. sale of the risky asset and hence the initial price $\tilde S _0$ is defined to match the buying resp.\ selling price of the risky asset.

\begin{proof}
The dynamics of $\tilde S_t$ result from It\^{o}'s formula, the dynamics of $\Upsilon_t$, and the identity
\begin{equation}\label{w2ableitung}
w''(y) = 2(\gamma-1)w'(y) w(y)- (2\mu/\sigma^2-1) w'(y),
\end{equation}
obtained by differentiating the ODE \eqref{eq:wode} for $w$ with respect to $y$. Therefore it remains to show that $\tilde{S}_t$ indeed takes values in the bid-ask spread $[(1-\ve)S_t,S_t]$. To this end, notice that -- in view of the ODE \eqref{eq:wode} for $w$ -- the derivative of the function $g(y):=w(y)/l e^y (1-w(y))$ is given by
$$g'(y)=\frac{w'(y)-w(y)+w^2(y)}{le^y (1-w(y))^2}=\frac{\gamma(w^2-2\frac{\mu}{\gamma\sigma^2} w)+(\mu^2-\lambda^2)/\gamma\sigma^4}{le^y (1-w(y))^2}.$$
Due to the boundary conditions for $w$, the function $g'$ vanishes at $0$ and $\log(u/l)$. Differentiating its numerator gives $2\gamma w'(y)(w(y)-\frac{\mu}{\gamma\sigma^2})$. For $\frac{\mu}{\gamma\sigma^2} \in (0,1)$ (resp.\ $\frac{\mu}{\gamma\sigma^2}>1$), $w$ is increasing from $\frac{\mu-\lambda}{\gamma\sigma^2}<\frac{\mu}{\gamma\sigma^2}$ to $\frac{\mu+\lambda}{\gamma\sigma^2}>\frac{\mu}{\gamma\sigma^2}$ on $[0,\log(u/l)]$ (resp.\ decreasing from $\frac{\mu+\lambda}{\gamma\sigma^2}$ to $\frac{\mu-\lambda}{\gamma\sigma^2}$ on $[\log(u/l),0]$); hence, $w'$ is nonnegative (resp.\ nonpositive). Moreover, $g'$ starts at zero for $y=0$ (resp.\ $\log(u/l)$), then decreases (resp.\ increases), and eventually starts increasing (resp.\ decreasing) again, until it reaches level zero again for $y=\log(u/l)$ (resp.\ $y=0$). In particular, $g'$ is nonpositive (resp.\ nonnegative), so that $g$ is decreasing on $[0,\log(u/l)]$ (resp.\ increasing on $[\log(u/l),0]$ for $\frac{\mu}{\gamma\sigma^2}>1$). Taking into account that $g(0)= 1$ and $g(\log(u/l))=1-\varepsilon$, by the boundary conditions for $w$ and the definition of $u$ and $l$ in Lemma \ref{lem:lambda}, the proof is now complete.
\end{proof}

\subsection{Verification}

The long-run optimal portfolio in the frictionless ``shadow market'' with price process $\tilde{S}_t$ can now be determined by adapting the argument in \citet*{guasoni.robertson.10}. The first step is to determine finite-horizon bounds, which provide upper and lower estimates for the maximal expected utility on any finite horizon $T$:

\begin{lemma}\label{lemfinite}
For a fixed time horizon $T>0$, let $\beta = \frac{\mu^2-\lambda^2}{2\gamma\sigma^2}$ and let the function $w$ be defined as in Lemma~\ref{lem:riccati}. Then, for the the shadow payoff $\tilde X_T$  corresponding to the risky fraction $\tilde \pi(\Upsilon_t) = w(\Upsilon_t)$ and the shadow discount factor $\tilde M_T=e^{-rT}\cale(-\int_0^\cdot \frac\tilmu\tilsigma dW_t)_T$ , the following bounds hold true: 
\begin{align}
E [\tilde X _T^{1-\gamma}] &= \tilde X_0 ^{1-\gamma} e^{(1-\gamma)(r+\beta)T}\hat{E}[e^{(1-\gamma)\left(\tilde q (\Upsilon_0)- \tilde q (\Upsilon_T) \right)}],\label{primbound}\\
E\left[\tilde M _T^{1-\frac{1}{\gamma}}\right]^{\gamma} &= e^{(1-\gamma)(r+\beta)T}\hat{E}\left[e^{(\frac{1}{\gamma}-1)\left(\tilde q (\Upsilon_0)- \tilde q (\Upsilon_T) \right)}\right]^{\gamma}, \label{dualbound}
\end{align}
where $\tilde q (y) := \int_0^y (w(z)-\frac{w'(z)}{1-w(z)}) dz$ and $\hat{E} \left[\cdot\right]$ denotes the expectation with respect to the \emph{myopic probability} $\hat{P}$, defined by
\begin{displaymath}
\frac{d \hat{P}}{dP} = \exp\left(\int_0^T \left(-\frac{\tilde \mu(\Upsilon_t)}{\tilde \sigma(\Upsilon_t)} + \tilde \sigma(\Upsilon_t) \tilde \pi(\Upsilon_t)\right) d W_t - \frac{1}{2}\int_{0}^T\left(-\frac{\tilde \mu(\Upsilon_t)}{\tilde \sigma(\Upsilon_t)}+ \tilde \sigma(\Upsilon_t) \tilde \pi(\Upsilon_t)\right)^2 d t\right).
\end{displaymath}
\end{lemma}

\begin{proof}
First note that $\tilmu, \tilsigma$, and $w$ are functions of $\Upsilon_t$, but the argument is omitted throughout to ease notation. Now, to prove \eqref{primbound}, notice that the frictionless shadow wealth process $\tilde X_t$ with dynamics $\frac{d\tilde{X}_t}{\tilde{X}_t}=w \frac{d\tilde{S}_t}{\tilde{S}_t}+(1-w)\frac{dS^0_t}{S^0_t}$ satisfies:
\begin{equation*}
\tilde X_T^{1-\gamma}=
\tilde{X}_0^{1-\gamma} e^{(1-\gamma)\int_0^T (r+\tilmu w -\frac{\tilsigma^2}{2}w^2) dt
+(1-\gamma)\int_0^T \tilsigma w dW_t}.
\end{equation*}
Hence:
\begin{align*}
\tilde X_T^{1-\gamma} =& \tilde{X}_0^{1-\gamma}\frac{d\hat P}{dP}
e^{\int_0^T ((1-\gamma)(r+\tilmu w -\frac{\tilsigma^2}{2}w^2)
+\frac12(-\frac{\tilmu}{\tilsigma}+\tilsigma w)^2)dt+\int_0^T ((1-\gamma)\tilsigma w-(-\frac{\tilmu}{\tilsigma}+\tilsigma w)) dW_t}.
\end{align*}
Inserting the definitions of $\tilmu$ and $\tilsigma$, the second integrand simplifies to $(1-\gamma)\sigma(\frac{w'}{1-w}-w)$. Similarly, the first integrand reduces to
$(1-\gamma)(r+\frac{\sigma^2}{2}(\frac{w'}{1-w})^2-(1-\gamma)\sigma^2 \frac{w' w}{1-w} +(1-\gamma)\frac{\sigma^2}{2}w^2)$. In summary:
\begin{equation}\label{secrep}
\tilde X_T^{1-\gamma} = \tilde{X}_0^{1-\gamma}\frac{d\hat P}{dP}
e^{(1-\gamma)\int_0^T (r+\frac{\sigma^2}{2}(\frac{w'}{1-w})^2-(1-\gamma)\sigma^2 \frac{w' w}{1-w} +(1-\gamma)\frac{\sigma^2}{2}w^2) dt + (1-\gamma)\int_0^T \sigma(\frac{w'}{1-w}-w) dW_t}.
\end{equation}
The boundary conditions for $w$ and  $w'$ imply $w(0)-\frac{w'(0)}{1-w(0)}=w(\log(u/l))-\frac{w'(\log(u/l))}{1-w(\log(u/l))}=0$; hence, It\^o's formula yields that the local time terms vanish in the dynamics of $\tilde{q}(\Upsilon_t)$:
\begin{equation}\label{eq:subs}
\tilde{q}(\Upsilon_T)-\tilde{q}(\Upsilon_0)=
\int_0^T \left(\mu-\tfrac{\sigma^2}{2}\right) \left(w-\tfrac{w'}{1-w}\right)+\tfrac{\sigma^2}{2}\left(w'-\tfrac{w''(1-w)+w'^2}{(1-w)^2}\right)dt+\int_0^T \sigma \left(w-\tfrac{w'}{1-w}\right)dW_t.  
\end{equation}
Substituting the second derivative $w''$ according to the ODE \eqref{w2ableitung} and using the resulting identity to replace the stochastic integral in \eqref{secrep} yields
\begin{align*}
\tilde X_T^{1-\gamma} =& \tilde{X}_0^{1-\gamma}\frac{d\hat P}{dP}
e^{(1-\gamma)\int_0^T (r+\frac{\sigma^2}{2}w'+(1-\gamma)\frac{\sigma^2}{2}w^2+(\mu-\frac{\sigma^2}{2})w)dt} e^{(1-\gamma)(\tilde{q}(\Upsilon_0)-\tilde{q}(\Upsilon_T))}.
\end{align*}
After inserting the ODE \eqref{eq:wode} for $w$, the first bound thus follows by taking the expectation.

The argument for the second bound is similar. Plugging in the definitions of $\tilmu$ and $\tilsigma$, the shadow discount factor $\tilde{M}_T=e^{-rT}\cale(-\int_0^\cdot \frac\tilmu\tilsigma dW)_T$ and the myopic probability $\hat P$ satisfy:
\begin{align*}
\tilde{M}_T^{1-\frac1\gamma}&=e^{\frac{1-\gamma}{\gamma}\int_0^T \frac{\tilmu}{\tilsigma}dW_t+\frac{1-\gamma}{\gamma}\int_0^T (r+\frac{\tilmu^2}{2\tilsigma^2})dt}\\
&=\frac{d\hat{P}}{dP} e^{\frac{1-\gamma}{\gamma}\int_0^T (\frac{\tilmu}{\tilsigma}-\frac{\gamma}{1-\gamma}(-\frac{\tilmu}{\tilsigma}+\tilsigma w))dW_t+\frac{1-\gamma}{\gamma}\int_0^T(r+\frac{\tilmu^2}{2\tilsigma^2}+\frac{\gamma}{2(1-\gamma)}(-\frac{\tilmu}{\tilsigma}+\tilsigma w)^2)dt}\\
&= \frac{d\hat{P}}{dP} e^{\frac{1-\gamma}{\gamma}\int_0^T \sigma(\frac{w'}{1-w}-w)dW_t+\frac{1-\gamma}{\gamma}\int_0^T(r+\frac{\sigma^2}{2}(\frac{w'}{1-w})^2-(1-\gamma)\sigma^2\frac{w' w}{1-w}+(1-\gamma)\frac{\sigma^2}{2}w^2)dt}.
\end{align*}
Again replace the stochastic integral using \eqref{eq:subs} and the ODE \eqref{w2ableitung}, obtaining
$$ \tilde{M}_T^{1-\frac1\gamma}=\frac{d\hat{P}}{dP} e^{\frac{1-\gamma}{\gamma}\int_0^T (r+\frac{\sigma^2}{2}w'+(1-\gamma)\frac{\sigma^2}{2}w^2+(\mu-\frac{\sigma^2}{2})w)dt}e^{\frac{1-\gamma}{\gamma}(\tilde{q}(\Upsilon_0)-\tilde{q}(\Upsilon_T))}.$$
Inserting the ODE \eqref{eq:wode} for $w$, taking the expectation, and raising it to power $\gamma$, the second bound follows.
\end{proof}

With the finite horizon bounds at hand, it is now straightforward to establish that the policy $\tilde{\pi}(\Upsilon_t)$ is indeed long-run optimal in the frictionless market with price $\tilde{S}_t$.

\begin{lemma}\label{lem:ergodic}
Let $0<\mu/\gamma \sigma^2 \neq 1$ and let $w$ be defined as in Lemma \ref{lem:riccati}.  Then, the risky weight $\tilde{\pi}(\Upsilon_t)=w(\Upsilon_t)$ is long-run optimal with equivalent safe rate $r+\beta$ in the frictionless market with price process $\tilde{S}_t$. The corresponding wealth process (in terms of $\tilde{S}_t$), and the numbers of safe and risky units are given by
\begin{align*}
\tilde{X}_t&=(\xi^0 S^0_0+\xi \tilde{S}_0)\mathcal{E}\left(\int_0^\cdot (r+w(\Upsilon_s) \tilde{\mu}(\Upsilon_s))ds+\int_0^\cdot w(\Upsilon_s)\tilde{\sigma}(\Upsilon_s)dW_s\right)_t, \\
\varphi_{0^-}&=\xi, \quad \varphi_t=w(\Upsilon_t)\tilde{X}_t/\tilde{S}_t \quad \mbox{for }t\geq 0,\\
\varphi^0_{0^-}&=\xi^0, \quad \varphi^0_t=(1-w(\Upsilon_t))\tilde{X}_t/S^0_t \quad \mbox{for }t\geq 0.
\end{align*}
\end{lemma}

\begin{proof}
The formulas for the wealth process and the corresponding numbers of safe and risky units follow directly from the standard frictionless definitions. Now let $\tilde{M}_t$ be the shadow discount factor from Lemma \ref{lemfinite}. Then, standard duality arguments for power utility (cf. Lemma~5 in \citet*{guasoni.robertson.10}) imply that the shadow payoff $\tilde{X}_t^\phi$ corresponding to \emph{any} admissible strategy $\phi_t$  satisfies the inequality
\begin{equation}\label{dualbound1}
\esp{(\tilde X^\phi_T)^{1-\gamma}}^{\frac 1{1-\gamma}}\le 
\esp{\tilde{M}_T^{\frac{\gamma-1}\gamma}}^{\frac\gamma{1-\gamma}} .
\end{equation}
This inequality in turn yields the following upper bound, valid for any admissible strategy $\phi_t$ in the frictionless market with shadow price $\tilde{S}_t$:
\begin{equation}\label{dualbound2}
\liminf_{T \to \infty} \frac{1}{(1-\gamma)T}\log E\left[(\tilde{X}^\phi_T)^{1-\gamma}\right] \le \liminf_{T\rightarrow\infty}
\frac{\gamma}{(1-\gamma)T}\log \esp{\tilde{M}_T^{\frac{\gamma-1}\gamma}}.
\end{equation}
Since the function $\tilde{q}$ is bounded on the compact support of $\Upsilon_t$, the second bound in Lemma \ref{lemfinite} implies that the right-hand side equals $r+\beta$. Likewise, the first bound in the same lemma implies that the shadow payoff $\tilde{X}_t$ (corresponding to the policy $\varphi_t$) attains this upper bound, concluding the proof.
\end{proof}

The next Lemma establishes that the candidate $\tilde{S}_t$ is indeed a shadow price. 

\begin{lemma}\label{lem:strategy}
Let $0<\mu/\gamma \sigma^2 \neq 1$. Then, the number of shares $\varphi_t=w(\Upsilon_t)\tilde{X}_t/\tilde{S}_t$ in the portfolio~$\tilde{\pi}(\Upsilon_t)$ in Lemma~\ref{lem:ergodic} has the dynamics  
\begin{equation}\label{eq:philu}
\frac{d\varphi_t}{\varphi_t}=\left(1-\frac{\mu-\lambda}{\gamma\sigma^2}\right)dL_t-\left(1-\frac{\mu+\lambda}{\gamma\sigma^2}\right)dU_t.
\end{equation}
Thus, $\varphi_t$ increases only when $\Upsilon_t=0$, that is, when $\tilde{S}_t$ equals the ask price, and decreases only when $\Upsilon_t=\log(u/l)$, that is, when $\tilde{S}_t$ equals the bid price.
\end{lemma} 

\begin{proof}
It\^o's formula and the ODE \eqref{w2ableitung} yield
$$dw(\Upsilon_t)=-(1-\gamma)\sigma^2 w'(\Upsilon_t)w(\Upsilon_t)dt+\sigma w'(\Upsilon_t)dW_t+w'(\Upsilon_t)(dL_t-dU_t).$$
Integrating $\varphi_t=w(\Upsilon_t)\tilde{X}_t/\tilde{S}_t$ by parts twice, inserting the dynamics of $w(\Upsilon_t)$, $\tilde{X}_t$, $\tilde{S}_t$, and simplifying yields:
$$\frac{d\varphi_t}{\varphi_t}=\frac{w'(\Upsilon_t)}{w(\Upsilon_t)}d(L_t-U_t).$$
Since $L_t$ and $U_t$ only increase (resp.\ decrease when $\mu/\gamma\sigma^2>1$) on $\{\Upsilon_t=0\}$ and $\{\Upsilon_t=\log(u/l)\}$, respectively, the assertion now follows from the boundary conditions for $w$ and $w'$.
\end{proof}

The optimal growth rate for {any} frictionless price within the bid-ask spread must be greater or equal than in the original market with bid-ask process $((1-\ve)S_t,S_t)$, because the investor trades at more favorable prices. For a \emph{shadow price}, there is an optimal strategy that only entails buying (resp.\ selling) stocks when $\tilde{S}$ coincides with the ask- resp.\ bid price. Hence, this strategy yields the same payoff when executed at  bid-ask prices, and thus is also optimal in the original model with transaction costs. The corresponding equivalent safe rate must also be the same, since the difference due to the liquidation costs vanishes as the horizon grows in~\eqref{eq:longrun}:

\begin{proposition}\label{prop:shadow}
For a sufficiently small spread $\ve$, the strategy $(\varphi^0_t,\varphi_t)$ from Lemma \ref{lem:ergodic} is also long-run optimal in the original market with transaction costs, with the same equivalent safe rate.
 \end{proposition}
 
 \begin{proof}
As $\varphi_t$ only increases (resp.\ decreases) when $\tilde{S}_t=S_t$ (resp.\ $\tilde{S}_t=(1-\ve)S_t$), the strategy $(\varphi_t^0,\varphi_t)$ is also self-financing for the bid-ask process $((1-\ve)S_t,S_t)$. Since $S_t \geq \tilde{S}_t \geq (1-\ve)S_t$ and the number $\varphi_t$ of risky shares is always positive, it follows that
\begin{equation}\label{eq:bounds}
\varphi^0_t S^0_t+\varphi_t \tilde{S}_t \geq \varphi^0_t S^0_t +\varphi_t^+(1-\ve)S_t-\varphi^-_t S_t  \geq (1-\tfrac{\ve}{1-\ve} \tilde{\pi}(Y_t))(\varphi^0_t S^0_t +\varphi_t \tilde{S}_t). 
\end{equation}
The shadow risky fraction $\tilde{\pi}(\Upsilon_t)=w(\Upsilon_t)$ is bounded from above by $(\mu+\lambda)/\gamma\sigma^2=\mu/\gamma\sigma^2+O(\ve^{1/3})$. For a sufficiently small spread $\ve$, the strategy $(\varphi_t^0,\varphi_t)$ is therefore also admissible for $((1-\ve)S_t,S_t)$. Moreover, \eqref{eq:bounds} then also yields
\begin{equation}\label{eq:optimalrate}
\begin{split}
\liminf_{T \to \infty} \frac{1}{(1-\gamma)T}\log E\left[(\varphi^0_T S^0_T+\varphi_T^+ (1-\ve)S_T -\varphi_T^- S_T)^{1-\gamma}\right]\\
\qquad = \liminf_{T \to \infty} \frac{1}{(1-\gamma)T}\log E\left[(\varphi^0_T S^0_T+\varphi_T\tilde{S}_T)^{1-\gamma}\right],
\end{split}
\end{equation}
that is, $(\varphi_t^0,\varphi_t)$ has the same growth rate, either with $\tilde{S}_t$ or with $[(1-\ve)S_t,S_t]$. 

For any admissible strategy $(\psi_t^0,\psi_t)$ for the bid-ask spread $[(1-\ve)S_t,S_t]$, set $\tilde{\psi}_t^0=\psi^0_{0^-}-\int_0^{t} \tilde{S}_s/S^0_s d\psi_s$. Then, $(\tilde{\psi}_t^0,\psi_t)$ is a self-financing trading strategy for $\tilde{S}_t$ with $\tilde{\psi}_t^0 \geq \psi_t^0$. Together with $\tilde{S}_t \in [(1-\ve)S_t,S_t]$, the long-run optimality of $(\varphi_t^0,\varphi_t)$ for $\tilde{S}_t$, and~\eqref{eq:optimalrate}, it follows that:
\begin{align*}
&\liminf_{T \to \infty} \frac{1}{T}\frac{1}{(1-\gamma)}\log E\left[(\psi^0_T S^0_T+\psi_T^+ (1-\ve)S_T -\psi_T^- S_T)^{1-\gamma}\right]\\
&\qquad \qquad\leq \liminf_{T \to \infty} \frac{1}{T}\frac{1}{(1-\gamma)}\log E\left[(\tilde{\psi}^0_T S^0_T+\psi_T \tilde{S}_T)^{1-\gamma}\right]\\
&\qquad \qquad\leq \liminf_{T \to \infty} \frac{1}{T}\frac{1}{(1-\gamma)}\log E\left[(\varphi^0_T S^0_T+\varphi_T \tilde{S}_T)^{1-\gamma}\right]\\
&\qquad \qquad= \liminf_{T \to \infty} \frac{1}{T}\frac{1}{(1-\gamma)}\log E\left[(\varphi^0_T S^0_T+\varphi_T^+ (1-\ve)S_T -\varphi_T^- S_T)^{1-\gamma}\right].
\end{align*}
Hence $(\varphi_t^0,\varphi_t)$ is also long-run optimal for $((1-\ve)S_t,S_t)$.
\end{proof}

The main result now follows by putting together the above statements.

\begin{theorem}\label{th:opt}
For $\ve>0$ small, and $0<\mu/\gamma\sigma^2 \neq 1$, the process $\tilde{S}_t$ in  Lemma~\ref{lem:dynamics} is a shadow price. A long-run optimal policy --- both for the frictionless market with price $\tilde{S}_t$ and in the market with bid-ask prices $(1-\ve)S_t,S_t$ --- is to keep the risky weight $\tilde{\pi}_t$ (in terms of $\tilde{S}_t$) in the no-trade region
$$[\pi_-,\pi_+]=\left[\frac{\mu-\lambda}{\gamma\sigma^2},\frac{\mu+\lambda}{\gamma\sigma^2}\right].$$
As $\ve \downarrow 0$, its boundaries have the asymptotics
\begin{align*}
  \pi_{\pm} = \frac{\mu}{\gamma\sigma^2}
    \pm \left(\frac{3}{4\gamma}\left(\frac{\mu}{\gamma\sigma^2}\right)^2\left(1-\frac{\mu}{\gamma\sigma^2}\right)^2\right)^{1/3} \ve^{1/3} 
    \pm \left(\frac{(5-2\gamma)}{10\gamma}\frac{\mu}{\gamma\sigma^2}\left(1-\frac{\mu}{\gamma\sigma^2}\right)^2-\frac{3}{20\gamma}\right) \ve
    +O(\ve^{4/3}).
\end{align*}
The corresponding equivalent safe rate is:
\[
  r+\beta=r+\frac{\mu^2-\lambda^2}{\gamma\sigma^2}=
    r+\frac{\mu^2}{2\gamma\sigma^2}-\frac{\gamma\sigma^2}{2}\left(\frac{3}{4\gamma} \left(\frac{\mu}{\gamma\sigma^2}\right)^2\left(1-\frac{\mu}{\gamma\sigma^2}\right)^2\right)^{2/3} \ve^{2/3}
    +O(\ve^{4/3}).
\]

If $\mu/\gamma\sigma^2=1$, then $\tilde{S}_t=S_t$ is a shadow price, and it is optimal to invest all wealth in the risky asset at time $t=0$, never to trade afterwards. In this case, the equivalent safe rate is given by the frictionless value $r+\beta=r+\mu^2/2\gamma\sigma^2$.
\end{theorem}

\begin{proof}
First let $0<\mu/\gamma\sigma^2 \neq 1$. Optimality with equivalent safe rate $r+\beta$ of the strategy $(\varphi_t^0,\varphi_t)$ associated to $\tilde{\pi}(\Upsilon_t)$ for $\tilde{S}_t$ has been shown in
Lemma~\ref{lem:ergodic}. The asymptotic expansions are an immediate consequence of the fractional power series for~$\lambda$ (cf.\ Lemma~\ref{lem:lambda}) and Taylor expansion. 

Next, Lemma~\ref{lem:strategy} shows that $\tilde{S}_t$ is a shadow price process in the sense of Definition~\ref{defi:shadow}. In view of the asymptotic expansions for $\pi_\pm$, Proposition~\ref{prop:shadow} shows that, for small transaction costs $\ve$, the same policy is also optimal, with the same equivalent safe rate, in the original market with bid-ask prices $(1-\ve)S_t,S_t$.

Consider now the degenerate case $\mu/\gamma\sigma^2=1$. Then the optimal strategy in the frictionless model $\tilde{S}_t=S_t$ transfers all wealth to the risky asset at time $t=0$, never to trade afterwards ($\varphi^0_t=0$ and $\varphi_t=\xi+\xi^0 S^0_0/S_0$ for all $t \geq 0$). Hence it is of finite variation and the number of shares never decreases, and increases only at time $t=0$, where the shadow price coincides with the ask price. Thus, $\tilde{S}_t=S_t$ is a shadow price. For small $\ve$, the remaining assertions then follow as in Proposition~\ref{prop:shadow} above.
\end{proof}

Next, the proof of Theorem \ref{th:finhor}, which establishes asymptotic finite-horizons bounds. In fact, the proof yields exact bounds in terms of $\lambda$, from which the expansions in the theorem are obtained.

\begin{proof}[Proof of Theorem \ref{lem:asympbounds}]
Let $(\phi^0,\phi)$ be any admissible strategy starting from the initial position $(\varphi^0_{0-},\varphi_{0-})$. Then as in the proof of Proposition~\ref{prop:shadow}, we have $\Xi^\phi_T \leq \tilde{X}^\phi_T$ for the corresponding shadow payoff, that is, the terminal value of the wealth process $\tilde{X}^\phi_t=\phi^0_0+\phi_0 \tilde{S}_0+\int_0^t \phi_s d\tilde{S}_s$ corresponding to trading $\phi$ in the frictionless market with price process $\tilde{S}_t$. Hence, Lemma 5 in \citet*{guasoni.robertson.10} and the second bound in Lemma \ref{lemfinite} imply that
\begin{equation}\label{eq:boundup}
\frac{1}{(1-\gamma)T}\log E\left[(\Xi^\phi_T)^{1-\gamma}\right]\leq r+\beta+\frac{1}{T}\log(\varphi^0_{0-}+\varphi_{0-}S_0)
+\frac{\gamma}{(1-\gamma)T}\log\hat{E}\left[e^{(\frac{1}{\gamma}-1)(\tilde{q}(\Upsilon_0)-\tilde{q}(\Upsilon_T))}\right].
\end{equation}
For the strategy $(\varphi^0,\varphi)$ from Lemma~\ref{lem:strategy}, we have $\Xi^\varphi_T \geq (1-\frac{\varepsilon}{1-\varepsilon}\frac{\mu+\lambda}{\gamma\sigma^2})\tilde{X}^\varphi_T$ by the proof of Proposition~\ref{prop:shadow}. Hence the first bound in Lemma \ref{lemfinite} yields
\begin{align}
\frac{1}{(1-\gamma)T}\log E\left[(\Xi^\varphi_T)^{1-\gamma}\right]\geq r+\beta +\frac{1}{T}\log(\varphi^0_{0-}+\varphi_{0-}\tilde{S}_0)
&+\frac{1}{(1-\gamma)T}\log\hat{E}\left[e^{(1-\gamma)(\tilde{q}(\Upsilon_0)-\tilde{q}(\Upsilon_T))}\right] \notag \\
&+\frac{1}{T}\log\left(1-\frac{\varepsilon}{1-\varepsilon}\frac{\mu+\lambda}{\gamma\sigma^2}\right).\label{eq:boundlow}
\end{align}
To determine explicit estimates for these bounds, we first analyze the sign of $\tilde{w}(y)=w-\frac{w'}{1-w}$ and hence the monotonicity of $\tilde{q}(y)=\int_0^y \tilde{w}(z)dz$. Whenever $\tilde{w}=0$, i.e., $w'=w(1-w)$, the derivative of $\tilde{w}$ is 
$$\tilde{w}'=w'-\frac{w''(1-w)+w'^2}{(1-w)^2}=\frac{(1-2\gamma)w' w+\frac{2\mu}{\sigma^2} w'}{1-w}-\left(\frac{w'}{1-w}\right)^2=2\gamma w\left(\frac{\mu}{\gamma\sigma^2}-w\right),$$
where we have used the ODE \eqref{w2ableitung} for the second equality. Since $\tilde{w}$ vanishes at $0$ and $\log(u/l)$ by the boundary conditions for $w$ and $w'$, this shows that the behaviour of $\tilde{w}$ depends on whether the investor's position is leveraged or not. In the absence of leverage, $\mu/\gamma\sigma^2 \in (0,1)$, $\tilde{w}$ is defined on $[0,\log(u/l)]$. It vanishes at the left boundary $0$ and then increases since its derivative is initially positive by the initial condition for $w$. Once the function $w$ has increased to level $\mu/\gamma\sigma^2§$, the derivative of $\tilde{w}$ starts to become negative; as a result, $\tilde{w}$ begins to decrease until it reaches level zero again at $\log(u/l)$. In particular, $\tilde{w}$ is nonnegative for $\mu/\gamma\sigma^2 \in (0,1)$. 

 In the leverage case $\mu/\gamma\sigma^2>1$, the situation is reversed. Then, $\tilde{w}$ is defined on $[\log(u/l),0]$ and, by the boundary condition for $w$ at $\log(u/l)$, therefore starts to decrease after starting from zero at $\log(u/l)$. Once $w$ has decreased to level $\mu/\gamma\sigma^2$, $\tilde{w}$ starts increasing until it reaches level zero again at $0$. Hence, $\tilde{w}$ is nonpositive for $\mu/\gamma\sigma^2>1$.

Now, consider Case 2 of Lemma \ref{lem:riccati}; the calculations for the other cases follow along the same lines with minor modifications. Then $\mu/\gamma\sigma^2 \in (0,1)$ and $\tilde{q}$ is positive and increasing. Hence, \begin{equation}\label{eq:boundup2}
\frac{\gamma}{(1-\gamma)T}\log\hat{E}\left[e^{(\frac{1}{\gamma}-1)(\tilde{q}(\Upsilon_0)-\tilde{q}(\Upsilon_T))}\right]\leq \frac{1}{T}\int_0^{\log(u/l)}\tilde{w}(y)dy
\end{equation}
and likewise
\begin{equation}\label{eq:boundlow2}
\frac{1}{(1-\gamma)T}\log\hat{E}\left[e^{(1-\gamma)(\tilde{q}(\Upsilon_0)-\tilde{q}(\Upsilon_T))}\right] \geq -\frac{1}{T}\int_0^{\log(u/l)} \tilde{w}(y)dy.
\end{equation}
Since $\tilde{w}(y)=w(y)-w'/(1-w)$, the boundary condions for $w$ imply
\begin{equation}\label{eq:expl3}
\int_0^{\log(u/l)} \tilde{w}(y)dy=\int_0^{\log(u/l)}w(y)dy -\log\left(\frac{\mu-\lambda-\gamma\sigma^2}{\mu+\lambda-\gamma\sigma^2}\right).
\end{equation}
By elementary integration of the explicit formula in Lemma \ref{lem:riccati} and using the boundary conditions from Lemma \ref{lem:smoothpasting} for the evaluation of the result at $0$ resp.\ $\log(u/l)$, the integral of $w$ can also be computed in closed form: 
\begin{align}\label{eq:expl4}
\int_0^{\log(u/l)} w(y)dy
=\tfrac{\frac{\mu}{\sigma^2}-\frac{1}{2}}{\gamma-1}\log\left(\tfrac{1}{1-\varepsilon}\tfrac{(\mu+\lambda)(\mu-\lambda-\gamma\sigma^2)}{(\mu-\lambda)(\mu+\lambda-\gamma\sigma^2)}\right)+\tfrac{1}{2(\gamma-1)}\log\left(\tfrac{(\mu+\lambda)(\mu+\lambda-\gamma\sigma^2)}{(\mu-\lambda)(\mu-\lambda-\gamma\sigma^2)}\right).
\end{align}
As $\epsilon \downarrow 0$, a Taylor expansion and the power series for $\lambda$ then yield
$$\int_0^{\log(u/l)} \tilde{w}(y)dy=\frac{\mu}{\gamma\sigma^2}\varepsilon+O(\varepsilon^{4/3}).$$
Likewise,
$$\log\left(1-\frac{\varepsilon}{1-\varepsilon}\frac{\mu-\lambda}{\gamma\sigma^2}\right)=-\frac{\mu}{\gamma\sigma^2}\varepsilon+O(\varepsilon^{4/3}),$$
as well as
$$
\log(\varphi^0_{0-}+\varphi_{0-}\tilde{S}_0) \geq \log(\varphi^0_{0-}+\varphi_{0-}S_0)- \frac{\varphi_{0-}S_0}{\varphi^0_{0-}+\varphi_{0-}S_0}\varepsilon+O(\varepsilon^2).$$
The claimed bounds then follow from \eqref{eq:boundup} and \eqref{eq:boundup2} resp.\ \eqref{eq:boundlow} and \eqref{eq:boundlow2}.
\end{proof}

\section{Trading Volume}

As above, let $\varphi_t=\varphi_t^{\uparrow}-\varphi_t^{\downarrow}$ denote the number of risky units at time $t$, written as the difference of the cumulated numbers of shares bought resp.\ sold until $t$. \emph{Relative share turnover}, defined as the measure $d\|\varphi\|_t/|\varphi_t|=d\varphi_t^{\uparrow}/|\varphi_t|+d\varphi^{\downarrow}_t/|\varphi_t|$, is a scale-invariant indicator of trading volume \citep*{lo2000trading}.
The \emph{long-term average share turnover} is defined as
$$
\lim_{T\rightarrow\infty}\frac1T\int_0^T \frac{d\|\varphi\|_t}{|\varphi_t|}.
$$
Similarly, \emph{relative wealth turnover} $(1-\ve)S_td\varphi^{\downarrow}_t/(\varphi^0_t S^0_t+\varphi_t (1-\ve)S_t)+S_t d\varphi^{\uparrow}_t/(\varphi^0_t S^0_t+\varphi_t S_t)$ is defined as the amount of wealth transacted divided by current wealth, where both quantities are evaluated in terms of the bid price $(1-\ve)S_t$ when selling shares resp.\ in terms of the ask price $S_t$ when purchasing them. As above, the \emph{long-term average wealth turnover} is then defined as
$$
\lim_{T\rightarrow\infty}\frac1T\left(\int_0^T \frac{(1-\ve)S_td\varphi^{\downarrow}_t}{\varphi^0_t S^0_t+\varphi_t (1-\ve)S_t}+\int_0^T\frac{S_t d\varphi^{\uparrow}_t}{\varphi^0_t S^0_t+\varphi_t S_t}\right).
$$

Both of these limits admit explicit formulas in terms of the gap, which yield asymptotic expansions for $\ve \downarrow 0$. The analysis starts with a preparatory result (cf.\ \citet*[Remark 4]{MR2048827} for the case of driftless Brownian motion).

\begin{lemma}\label{lem:lindif}
Let $\Upsilon_t$ be a diffusion on an interval $[l,u]$, $0<l<u$, reflected at the boundaries, i.e.
\begin{equation}
d\Upsilon_t = b(\Upsilon_t) dt + a(\Upsilon_t)^{1/2} dW_t + dL_t - dU_t,
\end{equation}
where the mappings $a(y)>0$ and $b(y)$ are both continuous, and the continuous, nondecreasing local time processes $L_t$ and $U_t$ satisfy $L_0=U_0=0$ and only increase on $\{L_t=l\}$ and $\{U_t=u\}$, respectively. Denoting by $\nu(y)$ the invariant density of $\Upsilon_t$, the following almost sure limits hold:
\begin{equation}
\lim_{T\rightarrow\infty} \frac{L_T}{T} = \frac{a(l) \nu(l)}2,
\qquad
\lim_{T\rightarrow\infty} \frac{U_T}{T} = \frac{a(u) \nu(u)}2.
\end{equation}
\end{lemma}

\begin{proof}
For $f\in C^2([l,u])$, write $\mathcal{L} f(y):=b(y) f'(y)+a(y)f''(y)/2$. Then, by It\^o's formula:
\[
\frac{f(\Upsilon_T)-f(\Upsilon_0)}{T} = \frac1T\int_0^T \mathcal{L}f(\Upsilon_t)dt
+\frac1T \int_0^T f'(\Upsilon_t) a(\Upsilon_t)^{1/2}dW_t +f'(l)\frac{L_T}T-f'(u)\frac{U_T}T.
\]
Now, take $f$ such that $f'(l)=1$ and $f'(u)=0$, and pass to the limit $T\rightarrow\infty$. The left-hand side vanishes because $f$ is bounded; the stochastic integral also vanishes by the Dambis-Dubins-Schwarz theorem, the law of the iterated logarithm, and the boundedness of $f'$. Thus, the ergodic theorem \citep*[II.35 and II.36]{MR1912205} implies that
\[
\lim_{T\rightarrow\infty} \frac{L_T}{T} = -\int_l^u \mathcal{L}f(y) \nu(y) dy.
\]
Now, the self-adjoint representation \citep*[VII.3.12]{MR1725357} $\mathcal{L}f = (a f' \nu)'/2 \nu$ yields:
\[
\lim_{T\rightarrow\infty} \frac{L_T}{T} = 
-\frac12\int_l^u (af'\nu)'(y)  dy =
\frac{a(l)\nu(l)f'(l)}2 - \frac{a(u)\nu(u)f'(u)}2 = \frac{a(l)\nu(l)}2.
\]
The other limit follows from the same argument, using $f$ such that $f'(l)=0$ and $f'(u)=1$.
\end{proof}

\begin{lemma}\label{lem:locallimits}
Let $0<\mu/\gamma\sigma^2 \neq 1$ and, as in~\eqref{eq:reflected}, let
$$\Upsilon_t=\left(\mu-\frac{\sigma^2}{2}\right)t+\sigma W_t +L_t-U_t$$
be Brownian motion with drift, reflected at $0$ and $\log(u/l)$. Then if $\mu \neq \sigma^2/2$, the following almost sure limits hold:
$$ 
 \lim_{T \to \infty} \frac{L_T}{T}= \frac{\sigma^2}{2}\left(\frac{\frac{2\mu}{\sigma^2}-1}{(u/l)^{\frac{2\mu}{\sigma^2}-1}-1}\right)
\quad \mbox{and} \quad 
\lim_{T \to \infty} \frac{U_T}{T}=\frac{\sigma^2}{2}\left(\frac{1-\frac{2\mu}{\sigma^2}}{(u/l)^{1-\frac{2\mu}{\sigma^2}}-1}\right) 
.
 $$
If $\mu=\sigma^2/2$, then $\lim_{T\to \infty}L_T/T=\lim_{T \to \infty} U_T/T=\sigma^2/(2\log(u/l))$ a.s.
 \end{lemma}

\begin{proof}[Proof]
First let $\mu \neq \sigma^2/2$. Moreover, suppose that $\mu/\gamma\sigma^2 \in (0,1)$. Then the scale function and the speed measure of the diffusion $\Upsilon_t$ are
\begin{align*}
s(y) =& 
\int_0^y \exp\Big(-2\int_0^\xi \frac{\mu-\frac{\sigma^2}{2}}{\sigma^2}d\zeta\Big)d\xi=\frac{1}{1-\frac{2\mu}{\sigma^2}} e^{(1-\frac{2\mu}{\sigma^2})y}, \\
m(dy) =&
1_{[0,\log(u/l)]}(y)\frac{2 dy}{s'(y)\sigma^2}=1_{[0,\log(u/l)]}(y) \frac{2}{\sigma^2}e^{(\frac{2\mu}{\sigma^2}-1)y}dy.
\end{align*}
The invariant distribution of $\Upsilon_t$ is the normalized speed measure
$$
\nu(dy)=\frac{m(dy)}{m([0,\log(u/l)])}=1_{[0,\log(u/l)]}(y)\frac{\frac{2\mu}{\sigma^2}-1}{(u/l)^{\frac{2\mu}{\sigma^2}-1}-1}e^{(\frac{2\mu}{\sigma^2}-1)y}dy.
$$
For $\mu/\gamma\sigma^2>1$, the endpoints $0$ and $\log(u/l)$ exchange their roles, and the result is the same, up to replacing $[0,\log(u/l)]$ with $[\log(u/l),0]$ and multiplying the formula by $-1$. Then, the claim follows from Lemma \ref{lem:lindif}.
In the case $\mu=\sigma^2/2$ of driftless Brownian motion, $\Upsilon_t$ has uniform stationary distribution on $[0,\log(u/l)]$ (resp.\ on $[\log(u/l),0]$ if $\mu/\gamma\sigma^2>1$), and the claim again follows by Lemma \ref{lem:lindif}.
\end{proof}

Lemma~\ref{lem:locallimits} and the formula for $\varphi_t$ from Lemma~\ref{lem:strategy} yield the long-term average trading volumes.
The asymptotic expansions then follow from the power series for~$\lambda$
(cf.\ Lemma~\ref{lem:lambda}).

\begin{corollary}\label{cor:trading}
If $\mu/\gamma\sigma^2 \neq 1$, the long-term average share turnover is 
\begin{align*}
\lim_{T\rightarrow\infty}\frac1T\int_0^T \frac{d\|\varphi\|_t}{|\varphi_t|}= \left(1-\frac{\mu-\lambda}{\gamma\sigma^2}\right)\lim_{T \to \infty} \frac{L_T}{T}+\left(1-\frac{\mu+\lambda}{\gamma\sigma^2}\right) \lim_{T \to \infty} \frac{U_T}{T},
\end{align*}
\nada{
and, as $\ve \downarrow 0$, has the asymptotic expansion
\[
 \lim_{T\rightarrow\infty}\frac1T\int_0^T \frac{d\|\varphi\|_t}{|\varphi_t|}= \frac{\sigma^2}{2}\frac{\mu}{\gamma\sigma^2}\left(1-\frac{\mu}{\gamma\sigma^2}\right)^2\left(\frac{3}{4\gamma} \left(\frac{\mu}{\gamma\sigma^2}\right)^2\left(1-\frac{\mu}{\gamma\sigma^2}\right)^2\right)^{-1/3}\ve^{-1/3}
  + O(\ve^{1/3}).
\]
}
and the long-term average wealth turnover is
\begin{align*}
&\lim_{T\to \infty} \frac{1}{T}\left(\int_0^T \frac{(1-\ve)S_td\varphi^{\downarrow}_t}{\varphi^0_tS^0_t+\varphi_t (1-\ve)S_t}+\int_0^T \frac{S_td\varphi^{\uparrow}_t}{\varphi^0_tS^0_t+\varphi_t S_t}\right)\\
&\qquad =\frac{\mu-\lambda}{\gamma\sigma^2}\left(1-\frac{\mu-\lambda}{\gamma\sigma^2}ö\right)\lim_{T \to \infty} \frac{L_T}{T}+\frac{\mu+\lambda}{\gamma\sigma^2}\left(1-\frac{\mu+\lambda}{\gamma\sigma^2}\right)\lim_{T \to \infty} \frac{U_T}{T},
\end{align*}
\nada{
and, as $\ve \downarrow 0$, has the asymptotic expansion
\begin{align*}
&\lim_{T\to \infty} \frac{1}{T}\left(\int_0^T \frac{(1-\ve)S_td\varphi^{\downarrow}_t}{\varphi^0_tS^0_t+\varphi_t (1-\ve)S_t}+\int_0^T \frac{S_t d\varphi^{\uparrow}_t}{\varphi^0_tS^0_t+\varphi_t S_t}\right)\\
&\qquad=\frac{\gamma\sigma^2}{3} \left(\frac{3}{4\gamma} \left(\frac{\mu}{\gamma\sigma^2}\right)^2\left(1-\frac{\mu}{\gamma\sigma^2}\right)^2\right)^{2/3} \ve^{-1/3}+O(\ve^{1/3}).
\end{align*}
Algorithmic calculations deliver terms of higher order.
}
If $\mu/\gamma\sigma^2=1$, the long-term average share and wealth turnover both vanish. 
\end{corollary}

\bibliographystyle{agsm}
\bibliography{tractrans}
\newpage

\end{document}